\documentclass[journal]{IEEEtran}
\usepackage{amsmath}
\usepackage{amssymb}
\usepackage{amsfonts}
\usepackage{graphicx}
\usepackage{epsfig}
\usepackage{subfigure}
\usepackage{psfrag}
\usepackage{cite}
\usepackage{latexsym}
\usepackage{url}
\usepackage{color}
\PassOptionsToPackage{bookmarks={false}}{hyperref}

\begin{document}
\title{Proactive Eavesdropping via Cognitive Jamming in Fading Channels}
\author{Jie Xu, Lingjie Duan, and Rui Zhang
\thanks{Part of this paper was presented in the IEEE International Conference on Communications (ICC), Kuala Lumpur, Malaysia, May 23-27, 2016.}
\thanks{J. Xu is with the School of Information Engineering, Guangdong University of Technology (e-mail: jiexu.ustc@gmail.com). He was with the Engineering Systems and Design Pillar, Singapore University of Technology and Design.}
\thanks{L. Duan is with the Engineering Systems and Design Pillar, Singapore University of Technology and Design (e-mail:~lingjie\_duan@sutd.edu.sg). L. Duan is the corresponding author.}
\thanks{R. Zhang is with the Department of Electrical and Computer Engineering, National University of Singapore (e-mail: elezhang@nus.edu.sg). He is also with the Institute for Infocomm Research, A*STAR, Singapore.}}

\maketitle

\begin{abstract}
To enhance the national security, there is a growing need for authorized parties to legitimately monitor suspicious communication links for preventing intended crimes and terror attacks. In this paper, we propose a new wireless information surveillance paradigm by investigating a scenario where a legitimate monitor aims to intercept a suspicious wireless link over fading channels. The legitimate monitor can successfully eavesdrop (decode) the information of the suspicious link at each fading state only when its achievable data rate is no smaller than that at the suspicious receiver. We propose a new approach, namely proactive eavesdropping via cognitive jamming, in which the legitimate monitor purposely jams the receiver in a full-duplex mode so as to change the suspicious communication (e.g., to a smaller data rate) for overhearing more efficiently. By assuming perfect self-interference (SI) cancelation (SIC) and global channel state information (CSI) at the legitimate monitor, we characterize the fundamental information-theoretic limits of proactive eavesdropping. We consider both delay-sensitive and delay-tolerant applications for the suspicious communication, under which the legitimate monitor maximizes the eavesdropping non-outage probability (for event-based monitoring) and the relative eavesdropping rate (for content analysis), respectively, by optimizing the jamming power allocation over different fading states subject to an average power constraint. Numerical results show that the proposed proactive eavesdropping via cognitive jamming approach greatly outperforms other benchmark schemes. Furthermore, by extending to a more practical scenario with residual SI and local CSI, we design an efficient {\emph{online}} cognitive jamming scheme inspired by the optimal cognitive jamming with perfect SIC and global CSI.
\end{abstract}


\newtheorem{definition}{\underline{Definition}}[section]
\newtheorem{fact}{Fact}
\newtheorem{assumption}{Assumption}
\newtheorem{theorem}{\underline{Theorem}}[section]
\newtheorem{lemma}{\underline{Lemma}}[section]
\newtheorem{corollary}{\underline{Corollary}}[section]
\newtheorem{proposition}{\underline{Proposition}}[section]
\newtheorem{example}{\underline{Example}}[section]
\newtheorem{remark}{\underline{Remark}}[section]
\newtheorem{algorithm}{\underline{Algorithm}}[section]
\newcommand{\mv}[1]{\mbox{\boldmath{$ #1 $}}}

\section{Introduction}

Recently, wireless security has attracted a lot of attentions from both academia and industry, and various approaches have been adopted to enhance the security of wireless networks among different layers of communication protocols \cite{ZouWangHanzo2015}. Among others, physical layer security techniques have been proposed as promising solutions to achieve perfect wireless secrecy against malicious eavesdropping attacks, and there are extensive studies in the literature investigating physical layer security techniques under different system setups (see, e.g., \cite{Wyner1975,GopalaLaiGamal2008,BlochBarrosRodriguesMcLaughlin2008,LiangPoorShamai2008} and the references therein). These existing works focus on preserving the confidentiality of wireless communications by assuming communication users to be rightful and viewing the information eavesdropping as malicious attacks. However, from a broader national security perspective, they overlook the possibility that communication links can also be used by criminals or terrorists and the resultant problems for information surveillance.

\begin{figure*}
\centering
 \epsfxsize=1\linewidth
    \includegraphics[width=15cm]{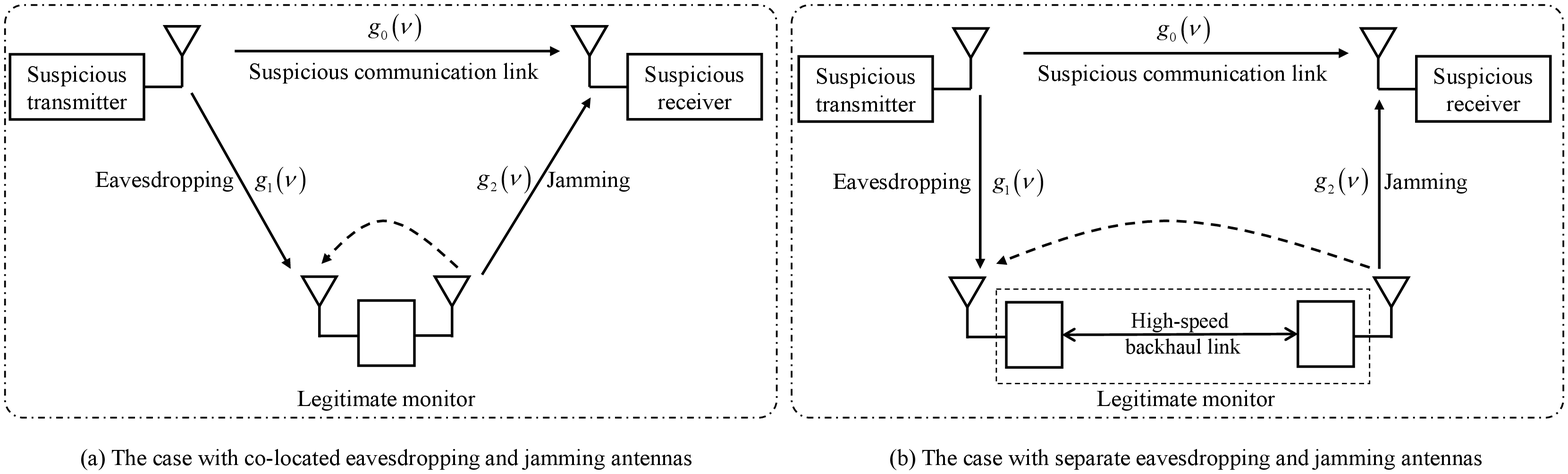}
\caption{An information surveillance scenario where a legitimate monitor proactively eavesdrops the suspicious communication from a transmitter to a receiver via cognitive jamming.} \label{fig:0}
\end{figure*}

With recent advancements in wireless technologies, many infrastructure-free wireless communication links are established for various applications. For example, smartphones in proximity can enable peer-to-peer data connections via Wi-Fi and bluetooth without Internet infrastructures,{\footnote{For example, FireChat is a mobile chatting software that allows nearby users to interconnect in a mobile ad hoc network by using Wi-Fi and/or Bluetooth locally (see {\url{https://www.technologyreview.com/s/525921/the-latest-chat-app-for-iphone-needs-no-internet-connection/}).}}} or via device-to-device (D2D) communications in the fifth-generation (5G) cellular networks without going through cellular infrastructures. Unmanned aerial vehicles (UAVs) can be employed as mobile relays to assist information exchange between ground users \cite{UAV1,UAV2,UAV3}. These emerging infrastructure-free wireless communications, however, can be used by criminals or terrorists to commit crimes or terror attacks. For instance, terrorists can use them to share information on a public transportation (e.g., in a plane) to facilitate hijacking or bombing activities, and undercover spies inside an isolated innovative enterprise can use them to send out the secret business data to outside peers. Since these communications do not go through any core infrastructures, they are difficult to be monitored by conventional surveillance approaches that intercept the communication data at the Internet backbones or cellular central offices.{\footnote{Note that the conventional approaches are used in the Terrorist Surveillance Program for legitimate information surveillance launched by the National Security Agency (NSA) of the United States \cite{TerroristWiKi}.}} As a result, there is a growing need for authorized parties (such as government agencies) to develop new wireless information surveillance approaches to legitimately monitor these infrastructure-free suspicious communication links (see, e.g., \cite{WireComMag,WCL,JSTSP,Spoofing1,Spoofing2}). These new wireless information surveillance approaches are also expected to be implemented to monitor infrastructure-based wireless communications in real time as a supplement of conventional Internet backbone surveillance.

To cope with the increasing information monitoring needs in wireless security, in this paper we propose a paradigm shift from the conventional physical layer security against {\it illegitimate eavesdropping} to the new information surveillance by exploiting {\it legitimate eavesdropping}. In particular, we consider a wireless scenario as shown in Fig. \ref{fig:0}, where a legitimate monitor aims to intercept a suspicious communication link from a transmitter to a receiver over fading channels.{\footnote{We assume that the suspicious transmitter and receiver have been detected {\it a priori} by authorized parties, and a legitimate monitor is assigned to monitor them accordingly. How to detect suspicious users and associate the suspicious users with the legitimate monitor can be referred to in \cite{WireComMag}.} Under this setup, the legitimate monitor can successfully eavesdrop (decode) the suspicious communication only when the received signal-to-noise ratio (SNR) (and accordingly the achievable data rate) at the legitimate monitor is no smaller than that at the suspicious receiver, since in this case the legitimate monitor is able to decode the whole information that can be decoded at the suspicious receiver.{\footnote{For the purpose of initial investigation, here we assume that the suspicious communication does not employ advanced anti-eavesdropping schemes such as the physical-layer security techniques.}} In practice, such legitimate eavesdropping is particularly challenging, since the legitimate monitor may be far away from the suspicious transmitter and cannot eavesdrop efficiently. This motivates us to design new methods to improve the legitimate eavesdropping performance in this work.

We propose a proactive eavesdropping via cognitive jamming approach (see Fig. \ref{fig:0}), in which the legitimate monitor operates in a full-duplex mode, and purposely sends jamming signals to interfere with the suspicious link, so as to decrease the achievable data rate at the suspicious receiver for overhearing more efficiently. For such a full-duplex legitimate monitor, its eavesdropping and jamming antennas can either be co-located or separately located, as shown in Figs. \ref{fig:0}-(a) and \ref{fig:0}-(b), respectively. The co-located structure can facilitate the joint design of eavesdropping and jamming, but may lead to severe self-interference (SI) from the jamming to the eavesdropping antennas. Due to the finite dynamic range of practical analog-to-digital converter (ADC), such SI is difficult to be cancelled perfectly, although it is recently reported that advanced analog and digital SIC schemes are able to achieve up to 110 dB SI reduction \cite{fullduplex}. In contrast, although the separate structure requires an extra low-latency backhaul link to connect the eavesdropping and jamming antennas to enable their joint operation, it effectively alleviates the SI problem by extending the distance between the transmitting/receiver antennas. Furthermore, the separate structure may have better eavesdropping and jamming performances by distributing the corresponding antennas in proximity of the suspicious transmitter and receiver, respectively. Also, it is more resilient to the anti-eavesdropping of the suspicious transmitter, since the separately located eavesdropping antenna is less  susceptible to get exposed. For both co-located and separate structures, to maximize the effectiveness of jamming for eavesdropping, it is important for the legitimate monitor to cognitively control the jamming power according to different fading states under its limited jamming power constraint. The main results of this paper are summarized as follows.

First, by assuming perfect SI cancelation (SIC) and global channel state information (CSI) at the legitimate monitor, we characterize the fundamental information-theoretic performance limits of proactive eavesdropping. In particular, we consider two different applications (i.e., delay-sensitive and delay-tolerant applications) for the suspicious communication, under which the legitimate monitor is interested in maximizing the eavesdropping non-outage probability and the relative eavesdropping rate (to the suspicious link's rate), respectively. Accordingly, we formulate two optimization problems for the legitimate monitor, by optimizing its jamming power allocation over different fading states subject to an average power constraint.

For the delay-sensitive applications, the eavesdropping non-outage probability maximization problem is shown to be irrespective of the transmit power allocation strategies at the suspicious transmitter, and we obtain the optimal cognitive jamming solution via the Lagrangian duality method. It is shown that the legitimate monitor jams only over the desired fading states of successful eavesdropping. For the delay-tolerant applications, the relative eavesdropping rate maximization problem depends critically on the power allocation strategies at the suspicious transmitter. In particular, we consider two commonly used transmit power allocation strategies (i.e., fixed power transmission and water-filling power allocation) at the suspicious transmitter, and obtain the optimal cognitive jamming solutions for the legitimate monitor. It is shown that the legitimate monitor may also jam over the undesired fading states of unsuccessful eavesdropping, since such jamming helps reduce the communication rate of the suspicious link in these fading states and therefore increase the percentage of successful eavesdropping rate in the desired fading states. Numerical results show that the proposed proactive eavesdropping via cognitive jamming approach greatly outperforms three benchmark schemes including the conventional passive eavesdropping without jamming, the proactive eavesdropping with constant-power jamming, and the proactive eavesdropping with on-off jamming.

Next, inspired by the above optimal cognitive jamming with perfect SIC and global CSI, we further design an {\it online} cognitive jamming scheme under practical assumptions of residual SI and local CSI. It is shown that the online cognitive jamming scheme achieves similar eavesdropping performance as the optimal cognitive jamming with perfect SIC and global CSI, especially when the legitimate monitor has separately equipped eavesdropping and jamming antennas.

It is worth noting that our proposed proactive eavesdropping via cognitive jamming approach is different from the conventionally investigated jamming and eavesdropping attacks in the literature. In particular, the conventional jamming has been investigated to disrupt wireless communications (e.g., of enemies in ballfields) without considering eavesdropping (see, e.g., \cite{KashyapBasarSrikant2004,LiuLiKongZhao2015}). In contrast, our paper utilizes jamming to facilitate the simultaneous eavesdropping at legitimate monitors. On the other hand, there have also been a handful of recent works investigating the secrecy capacity in the presence of active eavesdroppers that can both jam and eavesdrop \cite{MukherjeeSwindlehurst2011,ZhouMahamHjorungnes2011,KapetanovicZhengRusek2015,XiongLiangLiGong2015}. However, these existing works focused on preserving the confidentiality of wireless communications by viewing the (passive or active) eavesdropping as illegitimate attacks, while in this paper we look at a new research angle by considering eavesdropping as legitimate monitoring from the surveillance perspective.

The remainder of this paper is organized as follows. Section \ref{sec:II} presents the system model and formulates the eavesdropping non-outage probability and the relative eavesdropping rate maximization problems of our interest under perfect SIC and global CSI at the legitimate monitor. Section \ref{sec:III} develops the optimal solution to the eavesdropping non-outage probability maximization problem. Sections \ref{sec:IV} and \ref{sec:V} propose the optimal solutions to the relative eavesdropping rate maximization problems by considering that the suspicious transmitter adopts fixed power transmission and water-filling power allocation, respectively. Section \ref{sec:VI} shows the numerical results to validate the performance of our proposed proactive eavesdropping via cognitive jamming approach. Section \ref{sec:Pro:Eav} presents the online cognitive jamming scheme under a more practical scenario with residual SI and local CSI. Finally, Section \ref{sec:VII} concludes this paper.

\section{System Model and Problem Formulation}\label{sec:II}

As shown in Fig. \ref{fig:0}, we consider a point-to-point suspicious wireless communication link from a transmitter to a receiver over a frequency non-selective channel, and there is a legitimate monitor aiming to eavesdrop this link. The suspicious transmitter and receiver are each deployed with a single antenna, and the legitimate monitor is equipped with two antennas, one for eavesdropping (receiving) and the other for jamming (transmitting). The legitimate monitor can operate in a full-duplex mode to jam and eavesdrop at the same time. We consider a block fading model, where the wireless channels remain constant over each block and may change from one block to another. Let $h_0(\nu)$, $h_1(\nu)$, and $h_2(\nu)$ denote the channel coefficients from the suspicious transmitter to the suspicious receiver, from the suspicious transmitter to the eavesdropping antenna of the legitimate monitor, and from the jamming antenna of the legitimate monitor to the suspicious receiver, respectively, where $\nu$ denotes the joint fading state. The corresponding channel power gains are denoted as $g_0(\nu) = |h_0(\nu)|^2$, $g_1(\nu) = |h_1(\nu)|^2$, and $g_2(\nu) = |h_2(\nu)|^2$, respectively. Here, $g_0(\nu)$, $g_1(\nu)$, and $g_2(\nu)$ are assumed to be three random variables with a continuous joint probability density function (PDF) denoted by $\phi_\nu(g_0,g_1,g_2)$. Both the suspicious transmitter and receiver perfectly know the CSI of the suspicious channel (i.e., $g_0(\nu)$).

In order to characterize the fundamental information-theoretic performance limits of proactive eavesdropping, we make two following two assumptions. First, the legitimate monitor can perfectly cancel the SI from the jamming antenna to the eavesdropping antenna by using advanced analog and digital SIC schemes \cite{Sabharwal2014}. Note that the implementation of SIC requires the legitimate monitor to know the loop-back channel from the jamming to the eavesdropping antennas (via efficient channel estimation) \cite{fullduplex}. Next, the legitimate monitor perfectly knows the global CSI of suspicious, eavesdropping, and jamming channels (i.e., $g_0(\nu)$, $g_1(\nu)$, and $g_2(\nu)$) at each fading state $\nu$, as well as the joint PDF $\phi_\nu(g_0,g_1,g_2)$. Note that the global CSI assumption has been commonly made in the information-theoretic literature (see, e.g., the correlated jamming in \cite{KashyapBasarSrikant2004} and the cognitive radio in \cite{Devroye2006,JovicicViswanath2009}). We will consider the practical scenario with residual SI and local CSI in Section \ref{sec:Pro:Eav}.

Let the message sent by the suspicious transmitter and the jamming signal generated by the legitimate monitor be denoted by $s$ and $x$, respectively, both of which are assumed to be circularly symmetric complex Gaussian (CSCG) random variables with zero mean and unit variance. Note that transmitting CSCG signals at the suspicious transmitter is known to achieve the channel capacity subject to the CSCG noise, while using CSCG jamming signals is the best strategy for the legitimate monitor to degrade the suspicious communication when the suspicious transmitter uses CSCG signaling \cite{KashyapBasarSrikant2004}. We consider that at each fading state $\nu$, the suspicious transmitter employs the transmit power $p(\nu) > 0$, and the legitimate monitor cognitively adjusts its jamming power to $q(\nu) \ge 0$. Let $P > 0$ and $Q > 0$ denote the maximum average transmit and jamming power at the suspicious transmitter and the legitimate monitor, respectively. Thus we have
\begin{align}
\mathbb{E}_{\nu}(p(\nu)) &\le P, \label{eqn:P:sum}\\
\mathbb{E}_{\nu}(q(\nu)) &\le Q, \label{eqn:P2:con1}
\end{align}
where $\mathbb{E}_{\nu}(\cdot)$ denotes the expectation over the joint fading state $\nu$. Then, the received signals at the suspicious receiver and the eavesdropping antenna of the legitimate monitor are respectively denoted as
\begin{align}
y_0 = & \sqrt{p(\nu)} h_0(\nu)s + \sqrt{q(\nu)}h_2(\nu)x + n_0,\\
y_1 =& \sqrt{p(\nu)} h_1(\nu)s + n_1,
\end{align}
where $n_0$ and $n_1$ with zero mean and variances $\sigma^2_0$ and $\sigma^2_1$ denote the additive white Gaussian noises (AWGNs) at the suspicious receiver and the legitimate monitor, respectively. Accordingly, the signal-to-interference-plus-noise ratio (SINR) at the suspicious receiver and the SNR at the legitimate monitor receiver are respectively denoted as
\begin{align}
\gamma_{0}(\nu) & = \frac{g_0(\nu)p(\nu)}{g_2(\nu)q(\nu)+\sigma^2_0},\\
\gamma_{1}(\nu) & = \frac{g_1(\nu)p(\nu)}{\sigma^2_1}.\label{eqn:gamma1}
\end{align}
As a result, the achievable rates (in bps/Hz) of the suspicious link and the eavesdropping link in the fading state $\nu$ are respectively denoted as
\begin{align}
r_0(\nu) = & \log_2\left(1+\frac{g_0(\nu)p(\nu)}{g_2(\nu)q(\nu)+\sigma^2_0}\right),\label{eqn:r0}\\
r_1(\nu) = & \log_2\left(1+\frac{g_1(\nu)p(\nu)}{\sigma^2_1 }\right).\label{eqn:r1}
\end{align}

Based on the SINR $\gamma_{0}(\nu)$ at the suspicious receiver and the SNR $\gamma_{1}(\nu)$ at the legitimate monitor for one particular fading state $\nu$, we consider that the legitimate monitor can successfully eavesdrop the suspicious communication only when $\gamma_{1}(\nu)$ is no smaller than $\gamma_0(\nu)$ (i.e., $\gamma_{1}(\nu) \ge \gamma_{0}(\nu)$ or equivalently  $r_1(\nu) \ge r_0(\nu)$), since in this case the legitimate monitor can successfully decode the information sent in the suspicious link. Here, in order to focus our study on the physical layer perspective, we have ignored the possible encryption and decryption methods that can be employed at higher layers in the suspicious user communication. Therefore, we introduce the following indicator function to denote the event of successful eavesdropping at the legitimate monitor:
\begin{align}
X(\nu) = \left\{\begin{array}{ll}
1,& {\rm if} ~\gamma_1(\nu) \ge \gamma_0(\nu)\\
0, & {\rm otherwise},
\end{array}\right.\label{eqn:indictor}
\end{align}
where $X(\nu) = 1$ and $X(\nu) = 0$ indicate eavesdropping non-outage and outage events, respectively. Note that the indicator function $X(\nu)$ is irrespective of the transmit power $p(\nu)$ at the suspicious transmitter. Accordingly, we define the eavesdropping rate of the legitimate monitor at fading state $\nu$ as
\begin{align}
r(\nu) = r_0(\nu)X(\nu).\label{eqn:r:system}
\end{align}

The legitimate eavesdropping performance depends on different application scenarios for the suspicious communication. Specifically, we consider both delay-sensitive and delay-tolerant suspicious applications, and define the corresponding legitimate eavesdropping performance metrics as follows.

First, consider delay-sensitive applications, in which the suspicious transmitter adopts non-zero transmit power $p(\nu)$ at each fading state to deliver {\it event-based information} with strict delay constraints (e.g., real-time videos taken by its own camera), and the legitimate monitor aims to continuously track or monitor critical suspicious events. In this case, the delivered suspicious messages (e.g., the real-time video clips) in different fading states have the same significance to report and infer such series of ongoing events, although they may be with different data rates (e.g., different resolutions) due to the channel fading. Under such an event-based scenario, it is beneficial for the legitimate monitor to successfully eavesdrop over as many fading states as possible. As a result, we introduce the eavesdropping non-outage probability, given by $\mathbb{E}_{\nu}(X(\nu))$, as the event-based legitimate eavesdropping performance metric. Then, we aim to maximize the eavesdropping non-outage probability $\mathbb{E}_{\nu}(X(\nu))$ by optimizing the jamming power allocation $\{q(\nu)\}$ at the legitimate monitor subject to its average power constraint in (\ref{eqn:P2:con1}), for which the optimization problem is formulated as
\begin{align}
{\rm{(P1)}}:~\max_{\{q(\nu)\}}&~\mathbb{E}_{\nu}(X(\nu))\nonumber\\
\mathrm{s.t.}~&~ q(\nu) \ge 0, \forall \nu\label{eqn:P2:con2}\\
&~(\ref{eqn:P2:con1}).\nonumber
\end{align}
Since the eavesdropping non-outage probability $X(\nu)$ is irrespective of the transmit power $p(\nu)$ at the suspicious transmitter, it is evident that the optimal cognitive jamming solution to (P1) is independent of the power allocation strategies employed at the suspicious transmitter. Also note that problem (P1) is non-convex in general, since its objective function is not concave over the jamming power allocation $\{q(\nu)\}$. Despite the non-convexity, we will solve problem (P1) optimally in Section \ref{sec:III}.

Next, consider delay-tolerant applications, where the suspicious transmitter sends {\it content-based information} (such as data files) to the receiver and the monitor targets at data accumulation and content analysis. In this case, every transmitted bit may have the same significance to help content analysis, and it is thus desirable for the legitimate monitor to eavesdrop as many bits (relative to the sent bits) as possible. As a result, we use the relative eavesdropping rate, defined as the average eavesdropping rate over the average communication rate of the suspicious link, i.e., $\frac{\mathbb{E}_\nu(r(\nu))}{\mathbb{E}_\nu(r_0(\nu))} = \frac{\mathbb{E}_\nu(r_0(\nu)X(\nu))}{\mathbb{E}_\nu(r_0(\nu))}$, as the content-based legitimate eavesdropping performance criterion. In this case, the relative eavesdropping rate maximization problem for the legitimate monitor is formulated as
\begin{align}
\max_{\{q(\nu)\}}&~\frac{\mathbb{E}_\nu( r_0(\nu)X(\nu))}{\mathbb{E}_\nu( r_0(\nu))} \label{problem:delay:insenstive}\\
\mathrm{s.t.}~&~(\ref{eqn:P2:con1})~{\rm and}~(\ref{eqn:P2:con2}).\nonumber
\end{align}
Problem (\ref{problem:delay:insenstive}) is in general more challenging to be solved than (P1), which is due to the fact that the objective function in (\ref{problem:delay:insenstive}) is non-concave and depends on the transmit power $\{p(\nu)\}$ employed at the suspicious transmitter. It is difficult to solve problem (\ref{problem:delay:insenstive}) under general power allocations at the suspicious transmitter. As a result,  in Sections \ref{sec:IV} and \ref{sec:V} we will solve problem (\ref{problem:delay:insenstive}) under two commonly adopted transmission schemes for the suspicious transmitter, i.e., fixed power transmission and water-filling power allocation, respectively.

\section{Optimal Cognitive Jamming in Delay-Sensitive Suspicious Applications}\label{sec:III}

%\section{Optimal Cognitive Jamming Solutions}

%In this section, we provide optimal solutions to problems (P1) and (\ref{problem:delay:insenstive}).
%
%\vspace{-0.5em}
%\subsection{Optimal Solution to (P1) in Delay-Sensitive Applications}
%\vspace{-0em}

First, we consider problem (P1) to maximize the eavesdropping non-outage probability for event-driven monitoring in delay-sensitive suspicious applications. Although (P1) is non-convex in general, one can verify that it satisfies the time-sharing condition defined in \cite{YuLui2006}, as shown in the following lemma.
\begin{lemma}\label{lemma:1}
Let $\{q^a(\nu)\}$ and $\{q^b(\nu)\}$ denote the optimal solutions to (P1) under the average jamming power constraints $Q^a$ and $Q^b$, respectively. Then for any $0\le \theta \le 1$, there always exists a feasible solution $\{q^c(\nu)\}$ such that
\begin{align*}
{\mathbb{E}_\nu(X^c(\nu))} &\ge \theta{\mathbb{E}_\nu(X^a(\nu))} + (1-\theta){\mathbb{E}_\nu(X^b(\nu))}, \\
\mathbb{E}_{\nu}(q^c(\nu)) &\le \theta Q^a + (1-\theta) Q^b,
\end{align*}
where $\{X^i(\nu)\}$ denotes the corresponding $\{X(\nu)\}$ in (\ref{eqn:indictor}) under the given jamming power allocation $\{q^i(\nu)\}$, $i\in\{a,b,c\}$.
\end{lemma}
\begin{IEEEproof}
This lemma can be proved by using a similar approach as shown in \cite{YuLui2006}. Consider each fading state $\nu$ which happens over a certain amount of time. Then we can allocate the jamming power $q^c(\nu)$ to be $q^a(\nu)$ for a $\theta$ percentage of the time, and $q^{b}(\nu)$ for the remaining $1-\theta$ percentage of the time. Then it follows that $X^c(\nu) = \theta X^a(\nu)+(1-\theta) X^b(\nu)$ and $q^c(\nu) = \theta q^a(\nu) + (1-\theta) q^b(\nu)$. By coming all these fading states, we have ${\mathbb{E}_\nu(X^c(\nu))} = \theta {\mathbb{E}_\nu(X^a(\nu))} + (1-\theta) {\mathbb{E}_\nu(X^b(\nu))}$, and $\mathbb{E}_{\nu}(q^c(\nu)) = \theta \mathbb{E}_{\nu}(q^a(\nu)) + (1-\theta) \mathbb{E}_{\nu}(q^b(\nu))\le \theta Q^a + (1-\theta) Q^b$. This implies that the time-sharing condition stipulated in \cite{YuLui2006} is satisfied for problem (P1), and therefore, this lemma is verified.
\end{IEEEproof}

The time-sharing condition in Lemma \ref{lemma:1} ensures that strong duality or zero duality gap holds between (P1) and its Lagrange dual problem \cite[Theorem 1]{YuLui2006}.{\footnote{The strong duality between (P1) and and its Lagrange dual problem can also be verified by using the technique in \cite{LuoZhang2008}, which uses the Lyapunov theorem in functional analysis to prove the strong duality for a class of problems with ``continuous formulations''.}} Therefore, we can use the Lagrange duality method to solve problem (P1) optimally \cite{BoydVandenberghe2004}. The optimal solution to (P1) is obtained in the following proposition.

\begin{proposition}\label{proposition:P1}
The optimal cognitive jamming solution to (P1) is given as
\begin{align}
&q_1^*(\nu) = \nonumber\\
&
\left\{
\begin{array}{ll}
\left(\frac{g_0(\nu)}{g_1(\nu)}\sigma^2_1 - \sigma^2_0\right)\frac{1}{g_2(\nu)}, & {\rm if}~0 < \left(\frac{g_0(\nu)}{g_1(\nu)}\sigma^2_1 - \sigma^2_0\right)\frac{1}{g_2(\nu)} < \frac{1}{\lambda_1^*}, \\
0, &{\rm otherwise},
\end{array}
\right.\label{eqn:optimal:P1}
\end{align}
where $\lambda^*_1$ denotes the optimal dual variable associated with the average jamming power constraint in (\ref{eqn:P2:con1}). In particular, if $Q$ is sufficiently large with $\mathbb{E}_{\nu}\left(\left(\frac{g_0(\nu)}{g_1(\nu)}\sigma^2_1 - \sigma^2_0\right)\frac{1}{g_2(\nu)}\right) < Q$, it follows that $\lambda_1^*\to 0$ and
\begin{align}
q_1^*(\nu) = \left(\frac{g_0(\nu)}{g_1(\nu)}\sigma^2_1 - \sigma^2_0\right)\frac{1}{g_2(\nu)}, \forall \nu.\label{eqn:optimal:P1:case1}
\end{align}
Otherwise, $\lambda^*_1$ is set such that $\mathbb{E}_{\nu}(q_1^*(\nu)) = Q$.
\end{proposition}
\begin{IEEEproof}
See Appendix \ref{app:P1}.
\end{IEEEproof}

Note that the optimal jamming power allocation $\{q_1^*(\nu)\}$ depends on $\left\{\left(\frac{g_0(\nu)}{g_1(\nu)}\sigma^2_1 - \sigma^2_0\right)\frac{1}{g_2(\nu)}\right\}$. For a fading state $\nu$ with $\left(\frac{g_0(\nu)}{g_1(\nu)}\sigma^2_1 - \sigma^2_0\right)\frac{1}{g_2(\nu)} \le 0$, the monitor can already overhear from the transmitter successfully without jamming. Thus, it always holds that $\gamma_1(\nu) \ge \gamma_0(\nu)$ and $X(\nu) = 1$, and thus no jamming is required, i.e., $q_1^*(\nu) = 0$. For each of the other fading states, $\left(\frac{g_0(\nu)}{g_1(\nu)}\sigma^2_1 - \sigma^2_0\right)\frac{1}{g_2(\nu)} > 0$ denotes the required jamming power for the legitimate monitor to successfully eavesdrop the suspicious link, under which it holds that $\gamma_1(\nu) = \gamma_0(\nu)$. Among these fading states, the legitimate monitor selects to jam those with $\left(\frac{g_0(\nu)}{g_1(\nu)}\sigma^2_1 - \sigma^2_0\right)\frac{1}{g_2(\nu)}$ smaller than the threshold $\frac{1}{\lambda_1^*}$, so as to maximize the eavesdropping non-outage probability while satisfying the average jamming power constraint.
\vspace{-0em}

\section{Optimal Cognitive Jamming in Delay-Tolerant Suspicious Applications with Fixed Power Transmission}\label{sec:IV}

In this section, we consider problem (\ref{problem:delay:insenstive}) to maximize the relative eavesdropping rate for content-driven monitoring in delay-tolerant suspicious applications, where the suspicious transmitter employs fixed power transmission, i.e., $p(\nu) = P, \forall \nu$. Note that fixed power transmission is a commonly used strategy that is easy to implement at the transmitter, while we will consider the case with adaptive power transmission at the suspicious transmitter in Section \ref{sec:V}. With fixed power transmission, we rewrite the achievable rate $r_0(\nu)$ of the suspicious link in (\ref{eqn:r0}) as
\begin{align}
\bar r_0(\nu) = & \log_2\left(1+\frac{g_0(\nu)P}{g_2(\nu)q(\nu)+\sigma^2_0}\right).
\end{align}
As a result, the relative eavesdropping rate maximization problem (\ref{problem:delay:insenstive}) is reformulated as
\begin{align}
{\rm{(P2)}}:~\max_{\{q(\nu)\}}&~\frac{\mathbb{E}_\nu(\bar r_0(\nu)X(\nu))}{\mathbb{E}_\nu(\bar r_0(\nu))}\nonumber\\
\mathrm{s.t.}~~&~(\ref{eqn:P2:con1})~{\rm and}~(\ref{eqn:P2:con2}).\nonumber
\end{align}

In the following, we solve problem (P2) by first equivalently transforming it into solving a sequence of feasibility problems, and then using the Lagrange duality method to solve each feasibility problem.

First, we introduce an auxiliary variable $t$, and equivalently re-express problem (P2) as
\begin{align}
{\rm{(P2.1)}}:~\max_{\{q(\nu)\},t}~&~ t \nonumber\\
\mathrm{s.t.}~&~{\mathbb{E}_\nu(\bar r_0(\nu)X(\nu))} \ge t {\mathbb{E}_\nu(\bar r_0(\nu))}\label{eqn:P21:con1}\\
&~(\ref{eqn:P2:con1})~{\mathrm{and}}~(\ref{eqn:P2:con2}).\nonumber
%&~q(\nu) \le Q_{{\rm peak}}, \forall \nu.
\end{align}
Then, we show that the optimal solution to problem (P2.1) can be obtained by equivalently solving a sequence of feasibility problems each for a fixed $t$ and given by
\begin{align}
{\rm{(P2.2)}}:~\mathrm{find}~&~{\{q(\nu)\}} \nonumber\\
\mathrm{s.t.}~&~(\ref{eqn:P2:con1}),~(\ref{eqn:P2:con2})~{\mathrm{and}}~(\ref{eqn:P21:con1}).
\end{align}
Suppose that the optimal value of problem (P2.1) is denoted as $t^*$, where it must hold that $0 \le t^* \le 1$. If problem (P2.2) is feasible under a given $t$, then it follows that $t^* \ge t$; otherwise, $t^* < t$. Thus, by solving problem (P2.2) with different $t$'s and applying a simple bisection search over $t$, $t^*$ can be obtained for problem (P2.1). As a result, to obtain the optimal solution to (P2.1) and thus (P2), we only need to solve problem (P2.2) under any given $0\le t\le 1$.

Next, we focus on solving problem (P2.2) with any given $0\le t\le 1$. Despite that problem (P2.2) is still non-convex, it can be verified that strong duality or zero duality gap holds for (P2.2), since it satisfies the time-sharing condition \cite{YuLui2006}, which can be similarly shown as in Lemma \ref{lemma:1}. For this reason, in the following we check the feasibility of (P2.2) and obtain its optimal solution (when it is feasible) by making use of the Lagrange dual function of problem (P2.2).

Let the dual variables associated with the constraints in (\ref{eqn:P21:con1}) and (\ref{eqn:P2:con1}) be denoted by $\mu \ge 0$ and $\lambda \ge 0$, respectively. Then the partial Lagrangian of problem (P2.2) is denoted as
\begin{align}
&\mathcal{L}_2(\{q(\nu)\},\mu,\lambda)  \nonumber\\ = &
\mu\left({\mathbb{E}_\nu((X(\nu)-t)\bar r_0(\nu))}\right)  - \lambda \left(\mathbb{E}_{\nu}(q(\nu)) - Q\right).
\end{align}
As a result, the dual function of (P2.2) is expressed as
\begin{align}
f_2(\mu,\lambda)=\max_{\{q(\nu)\ge 0\}}~\mathcal{L}_2(\{q(\nu)\},\mu,\lambda).\label{eqn:P2:dual:function}
\end{align}
Then, the following proposition helps determine whether problem (P2.2) is feasible or not.

\begin{proposition}\label{proposition:P2:1}
Problem (P2.2) is infeasible if and only if there exist $\mu \ge 0$ and $\lambda \ge 0$ such that $f_2(\mu,\lambda) < 0$.
\end{proposition}
\begin{IEEEproof}
See Appendix \ref{app:A}.
\end{IEEEproof}

Note that for any $\alpha > 0$, $f_2(\alpha\mu,\alpha\lambda) = \alpha f_2(\mu,\lambda)$. As a result, if problem (P2.2) is infeasible, then $f_2(\mu,\lambda)$ is unbounded from below, i.e., $f_2(\mu,\lambda) \to -\infty$; while if problem (P2.2) is feasible, then it follows that $\min_{\mu\ge 0,\lambda\ge 0}f_2(\mu,\lambda) = f_2(\mu_2^\star,\lambda_2^\star) = 0$ with $\mu_2^\star \ge 0$ and $\lambda_2^\star\ge 0$ being the optimal dual solutions to problem (P2.2). This observation will be used to develop a numerical algorithm to solve problem (P2.2) later.

Now, it remains to solve problem (\ref{eqn:P2:dual:function}) to obtain $f_2(\mu,\lambda)$ under any given $\mu \ge 0$ and $\lambda \ge 0$. By dropping the constant $\lambda Q$, problem (\ref{eqn:P2:dual:function}) can be decomposed into various subproblems as follows each for one fading state $\nu$.
\begin{align}
\max_{q(\nu)\ge 0}~\mu(X(\nu)-t)\bar r_0(\nu) - \lambda q(\nu)\label{eqn:P2:dual:function:decom}
\end{align}
We then have the following proposition.

\begin{proposition}\label{proposition:P2:2}
Under any given $\mu \ge 0$ and $\lambda \ge 0$, the optimal solution to problem (\ref{eqn:P2:dual:function:decom}) and thus problem (\ref{eqn:P2:dual:function}) is given as
\begin{align}
&q_2^{(\mu,\lambda)}(\nu) = \nonumber\\&
\left\{
\begin{array}{ll}
0, & {\rm if}~\left(\frac{g_0(\nu)}{g_1(\nu)}\sigma^2_1 - \sigma^2_0\right)\frac{1}{g_2(\nu)} \le 0\\
\left(\frac{g_0(\nu)}{g_1(\nu)}\sigma^2_1 - \sigma^2_0\right)\frac{1}{g_2(\nu)}, & {\rm if}~\left(\frac{g_0(\nu)}{g_1(\nu)}\sigma^2_1 - \sigma^2_0\right)\frac{1}{g_2(\nu)} > 0~\\
&~~~~~~{\rm and}~\bar v_1(\nu) \ge  \bar v_2(\nu) \\
\bar q(\nu),& {\rm if}~\left(\frac{g_0(\nu)}{g_1(\nu)}\sigma^2_1 - \sigma^2_0\right)\frac{1}{g_2(\nu)} > 0~\\
&~~~~~~{\rm and}~\bar v_1(\nu) <  \bar v_2(\nu),
\end{array}
\right.\label{eqn:dual:function:P2:solution}
\end{align}
with
\begin{align}
\bar q(\nu) \triangleq&
\min\bigg(\max\bigg(0,\frac{\sqrt{g_0^2(\nu)P^2 + 4 t \mu g_0(\nu)g_2(\nu) P/(\ln2\cdot \lambda)}}{2g_2(\nu)}\nonumber\\&
-\frac{g_0(\nu)P}{2g_2(\nu)} - \frac{\sigma_0^2}{g_2(\nu)}\bigg), \left(\frac{g_0(\nu)}{g_1(\nu)}\sigma^2_1 - \sigma^2_0\right)\frac{1}{g_2(\nu)}\bigg)\label{eqn:hatq}
\end{align}
denoting the optimal jamming power when the legitimate receiver cannot eavesdrop the suspicious link. Here, $\bar v_1(\nu)$ and $\bar v_2(\nu)$ denote the optimal values achieved by problem (\ref{eqn:P2:dual:function:decom}) when $X(\nu) = 1$ (eavesdropping is successful) and $X(\nu) = 0$ (eavesdropping is not successful), respectively, and are given by
\begin{align}
\bar v_1(\nu) = &\mu\left(1- t\right) \log_2\left(1+\frac{g_1(\nu)P}{\sigma^2_1 }\right) \nonumber\\&
- \lambda \left(\frac{g_0(\nu)}{g_1(\nu)}\sigma^2_1 - \sigma^2_0\right)\frac{1}{g_2(\nu)}, \label{eqn:v1}\\
\bar v_2(\nu) = & -\mu t \log_2\left(1+\frac{g_0(\nu)P}{g_2(\nu)\bar q(\nu) +\sigma^2_0}\right) - \lambda \bar q(\nu).\label{eqn:v2}
\end{align}
\end{proposition}

\begin{IEEEproof}
See Appendix \ref{app:B}.
\end{IEEEproof}

With Propositions \ref{proposition:P2:1} and \ref{proposition:P2:2} at hand, we are ready to present the complete algorithm to solve the feasibility problem (P2.2) via the subgradient based ellipsoid method \cite{BoydConvexII}, by using the fact that the subgradient of $f_2(\mu,\lambda)$ is %denoted as \begin{align*}
$\mv s_2(\mu,\lambda) = \left[\mathbb{E}_\nu\left(\left(X^{(\mu,\lambda)}(\nu) - t\right) \bar r_0^{(\mu,\lambda)}(\nu)\right),%\right. \nonumber\\ &\left.
Q - {\mathbb{E}_\nu(q_2^{(\mu,\lambda)}(\nu))}\right]^T$
%\end{align*}
under given $\mu$ and $\lambda$. Here, $\{\bar r_0^{(\mu,\lambda)}(\nu)\}$ and $\{X^{(\mu,\lambda)}(\nu)\}$ denote the corresponding $\{\bar r_0(\nu)\}$ and $\{X(\nu)\}$ under given $\{q_2^{(\mu,\lambda)}(\nu)\}$, respectively. The detailed algorithm for solving problem (P2.2) is summarized as Algorithm 1 in Table \ref{table:1}, for which the optimal solution (when (P2.2) is feasible) is denoted as $\left\{q_2^{(\mu_2^\star,\lambda_2^\star)}(\nu)\right\}$ with $\mu_2^\star$ and $\lambda_2^\star$ denoting the optimal dual solution.

\begin{table}[!t]\scriptsize
\caption{Algorithm for Solving the Feasibility Problem (P2.2)}
\label{table:framework} \centering
\begin{tabular}{|p{8cm}|}
\hline
\textbf{Algorithm 1}\\
\hline\vspace{0.01cm}
1) {\bf Initialization:} Set the iteration index $n=0$, and given an ellipsoid $\xi^{(0)} \subseteq \mathbb{R}^2$ centered at $[\mu^{(0)},\lambda^{(0)}]^T$. \\
2) {\bf Repeat:}
  \begin{itemize} \setlength{\itemsep}{0pt}
    \item[a)] Solve problem (\ref{eqn:P2:dual:function}) under given $\mu^{(n)}$ and $\lambda^{(n)}$ by using Proposition \ref{proposition:P2:2} to obtain $f_2(\mu^{(n)},\lambda^{(n)})$;
    \item[b)] If $f_2(\mu^{(n)},\lambda^{(n)}) < 0$, then problem (P2.2) is infeasible, exit the algorithm; otherwise, go to the next step.
    \item[c)] Update the ellipsoid $\xi^{(n+1)}$ based on $\xi^{(n)}$ and the subgradient $\mv s_2(\mu^{(n)},\lambda^{(n)})$. Set $[\mu^{(n+1)},\lambda^{(n+1)}]^T$ as the center for $\xi^{(n+1)}$.
    \item[d)] $n \gets n+1$;
  \end{itemize}
  3) {\bf Until} the stopping criteria for the ellipsoid method is met.\\
  4) {\bf Set} $\mu_2^\star = \mu^{(n)}$ and $\lambda_2^\star = \lambda^{(n)}$. Problem (P2.2) is feasible, and $\left\{q_2^{(\mu_2^\star,\lambda_2^\star)}(\nu)\right\}$ in Proposition \ref{proposition:P2:2} becomes its optimal solution.\\
 \hline
\end{tabular}\label{table:1}%\vspace{-2em}
\end{table}

Finally, by applying Algorithm 1 to solve problem (P2.2) together with the bisection search for finding the optimal $t^*$, we obtain the optimal solution to problem (P2.1) and (P2). Under $t^*$, denote the corresponding optimal dual solution $\mu_2^\star$ and $\lambda_2^\star$ to (P2.2) as $\mu_2^*$ and $\lambda_2^*$. Then the optimal solution to (P2.1) and (P2) is given as $\left\{q_2^*(\nu)\right\}$ with $q_2^*(\nu) = q_2^{(\mu_2^*,\lambda_2^*)}(\nu), \forall \nu$.

\section{Optimal Cognitive Jamming in Delay-Tolerant Suspicious Applications with Water-Filling Power Allocation}\label{sec:V}

In this section, we consider problem (\ref{problem:delay:insenstive}) to maximize the relative eavesdropping rate for content-driven monitoring in delay-tolerant suspicious applications, where the suspicious transmitter employs adaptive power transmission to maximize its own average communication rate via water-filling power allocation over different fading states. In this case, the power allocation $\{p(\nu)\}$ at the suspicious transmitter varies depending on the jamming power profile $\{q(\nu)\}$ at the legitimate monitor. As a result, we first present the water-filling power allocation at the suspicious transmitter under any given $\{q(\nu)\}$, and then present the relative eavesdropping rate maximization problem over $\{q(\nu)\}$ under such a power adaptation strategy.

First, suppose that the jamming power profile $\{q(\nu)\}$ at the legitimate monitor is given. In this case, the suspicious transmitter optimizes its transmit power allocation $\{p(\nu)\}$ to maximize its average achievable data rate $\mathbb{E}_\nu\left(r_0(\nu)\right)$ in (\ref{eqn:r0}) subject to its average power constraint in (\ref{eqn:P:sum}),
%\begin{align}
%\max_{\{p(\nu) \ge 0\}} ~& \mathbb{E}_\nu\left(r_0(\nu)\right)\nonumber\\
%\mathrm{s.t.}~~& (\ref{eqn:P:sum}), \nonumber
%\end{align}
for which the optimal water-filling power allocation solution is given as \cite{Goldsmith1997}{\footnote{The implementation of the water-filling power allocation requires the suspicious transmitter to know the interference power $g_2(\nu)q(\nu)$, which can be measured by the suspicious receiver and sent back to the suspicious transmitter.}
\begin{align}\label{eqn:solution:p}
\hat p(\nu) = \left[\frac{1}{\ln 2 \cdot \beta} - \frac{g_2(\nu)q(\nu)+\sigma^2_0}{g_0(\nu)} \right]^+,~\forall \nu,
\end{align}
where $[x]^+ = \max(x,0)$, and $\beta \ge 0$ is the Lagrange dual variable associated with the average transmit power constraint in (\ref{eqn:P:sum}) at the suspicious transmitter, such that
\begin{align}
\mathbb{E}_\nu\left(\hat p(\nu)\right)=P. \label{eqn:Pave}
\end{align}
Here, $\frac{1}{\ln 2 \cdot \beta}$ can be interpreted as the water level. Consequently, the resulting achievable rate of the suspicious link for the fading state $\nu$ is given by
\begin{align}
\hat r_0(\nu) %&= \log_2\left(1+\frac{g_0(\nu)p^{\star\star}(\nu)}{g_2(\nu)q(\nu) + \sigma^2_0}\right) \nonumber\\
%& = \left[\log_2\left( \frac{g_0(\nu)}{\ln 2 \cdot \beta(g_2(\nu)q(\nu) + \sigma^2_0)} \right)\right]^+\nonumber\\
& = \left[\log_2\left( \frac{g_0(\nu)}{\ln 2 \cdot \beta(g_2(\nu)q(\nu) + \sigma^2_0)} \right)\right]^+,\label{eqn:hatr}
\end{align}
and the corresponding relative eavesdropping rate is given as $\frac{\mathbb{E}_\nu(\hat r_0(\nu)X(\nu))}{\mathbb{E}_\nu(\hat r_0(\nu))}$.

Next, under the water-filling power allocation in (\ref{eqn:solution:p}) for the suspicious transmitter, the relative eavesdropping rate maximization problem (\ref{problem:delay:insenstive}) over the jamming power $\{q(\nu)\}$ for the legitimate monitor is re-expressed as
\begin{align}
{\rm{(P3)}}:~\max_{\{q(\nu)\},\beta \ge 0}&~\frac{\mathbb{E}_\nu(\hat r_0(\nu)X(\nu))}{\mathbb{E}_\nu(\hat r_0(\nu))}\nonumber\\
\mathrm{s.t.}~~&~(\ref{eqn:P2:con1}),~(\ref{eqn:P2:con2}),~{\rm and}~(\ref{eqn:Pave}),\nonumber
\end{align}
where $\beta$ is an auxiliary variable to be optimized in addition to $\{q(\nu)\}$.\footnote{It is worth noting that the cognitive jamming optimization problem (P3) here can be equivalently formulated as a bi-level optimization problem, where the lower-level optimization task is for the suspicious transmitter to maximize its average achievable rate $\mathbb{E}_{\nu}(r_0(\nu))$, and the upper-level optimization task is for the legitimate monitor to maximize the relative eavesdropping rate. In particular, (P3) is equivalent to the upper-level optimization task, while the water-filling power allocation in (\ref{eqn:solution:p})  corresponds to the optimal solution (with respect to an auxiliary variable $\beta$) to the lower-level optimization.} In the rest of this section, we focus on solving problem (P3).

Note that problem (P3) is a more difficult problem than (P2). This is due to the fact that both the objective function of (P3) and the constraint in (\ref{eqn:Pave}) are non-convex, and furthermore, the  auxiliary variable $\beta$ is related to all the fading states. To optimally solve problem (P3), we adopt an approach by first finding the optimal $\{q(\nu)\}$ under any given auxiliary variable $\beta$ that is feasible (i.e., for (\ref{eqn:P2:con1}),~(\ref{eqn:P2:con2}),~{\rm and}~(\ref{eqn:Pave}) to be satisfied at the same time), and then using a one-dimensional exhaustive search to obtain the optimal $\beta$ for (P3) over its feasible regime, i.e., the regime of $\beta$ for problem (P3) to be feasible. In the following, we first determine the feasible regime of $\beta$, and then optimize $\{q(\nu)\}$ for problem (P3) under any given $\beta$ within such a feasible regime.

%In the following, we first find the feasible regime of the auxiliary variable $\beta$, and then optimize $\{q(\nu)\}$ for problem (P3) under any given $\beta$ within such a feasible regime.

\subsection{Finding the Feasible Regime of $\beta$}\label{sec:VA}

%there exist upper and lower bounds of $\beta$, denoted by $\beta^{\max}$ and $\beta^{\min}$, respectively, such that problem (P3) is feasible when $\beta^{\min} \le \beta \le \beta^{\max}$, while it is not feasible when $\beta > \beta^{\max}$ or $\beta < \beta^{\min}$.

It is evident that in order for (\ref{eqn:P2:con1}),~(\ref{eqn:P2:con2}),~{\rm and}~(\ref{eqn:Pave}) to be satisfied at the same time, the feasible $\beta$ is upper and lower bounded by $\beta^{\max}$ and $\beta^{\min}$, respectively. First, we obtain the upper bound of $\beta$, i.e., $\beta^{\max}$. It is observed from (\ref{eqn:solution:p}) and (\ref{eqn:Pave}) that as the jamming power increases, the variable $\beta$ decreases accordingly. As a result, the upper bound of $\beta$ is achieved when the legitimate monitor does not send any jamming signals by setting $q(\nu) = 0,\forall \nu$. By using this together with (\ref{eqn:solution:p}) and (\ref{eqn:Pave}), the upper bound $\beta^{\max}$ can be obtained.

Next, we obtain the lower bound of $\beta$, i.e., $\beta^{\min}$. Based on the similar observation above, the lower bound $\beta^{\min}$ is achieved when full jamming power is employed with $\mathbb{E}_{\nu}\left(q(\nu)\right) = Q$. However, it remains unknown how the jamming power is allocated over different fading states. To overcome this issue, we propose to solve a series of feasibility problems each with a given $\beta$.
\begin{align}
{\rm find}~&~\{q(\nu)\} \label{eqn:29:feasibility}\\
\mathrm{s.t.}~&~(\ref{eqn:P2:con1}),~(\ref{eqn:P2:con2}),~{\rm and}~(\ref{eqn:Pave}).\nonumber
\end{align}
For any given $\beta$, if problem (\ref{eqn:29:feasibility}) is feasible, then $\beta^{\min} \le \beta$; otherwise, we have $\beta^{\min} > \beta$. Based on this observation, $\beta^{\min}$ can be found by solving problem (\ref{eqn:29:feasibility}) under any given $\beta$, together with a bisection search over $\beta$. Since it is known that $\beta$ should lie within the interval $[0, \beta^{\max}]$, the bisection search is employed over such an interval. Now, we only need to focus on solving problem (\ref{eqn:29:feasibility}) under any given $\beta \in [0, \beta^{\max}]$.

First, we show that strong duality holds for problem (\ref{eqn:29:feasibility}), although it is non-convex in general due to the nonlinear equality constraint in (\ref{eqn:Pave}).
\begin{lemma}\label{lemma:2}
Strong duality holds between problem (\ref{eqn:29:feasibility}) and its dual problem.
\end{lemma}
\begin{IEEEproof}
Note that the equality constraint in (\ref{eqn:Pave}) is indeed equivalent to two inequality constraints, i.e., $\mathbb{E}_\nu\left(\hat p(\nu)\right) \le P$ and $\mathbb{E}_\nu\left(\hat p(\nu)\right)\ge P$. As a result, the time sharing property still holds for problem (\ref{eqn:29:feasibility}) \cite{YuLui2006}. Therefore, strong duality holds between problem (\ref{eqn:29:feasibility}) and its dual problem. As a result, this lemma is proved.
\end{IEEEproof}

Next, we use the Lagrange duality method to check the feasibility of problem (\ref{eqn:29:feasibility}). Since the derivation procedure is similar to that for checking the feasibility of problem (P2.2), we omit the detail here, and leave it in Appendix \ref{app:D}.

\subsection{Optimizing $\{q(\nu)\}$ for Problem (P3) Under Any Given Feasible $\beta$}\label{sec:VB}

Next, we obtain the optimal cognitive jamming power solution $\{q(\nu)\}$ for problem (P3) under any given $\beta$ with $\beta^{\min}\le \beta \le \beta^{\max}$, for which the optimization problem is rewritten as:
\begin{align}
{\rm (P3.1)}:\max_{\{q(\nu)\}}&~\frac{\mathbb{E}_\nu(\hat r_0(\nu)X(\nu))}{\mathbb{E}_\nu(\hat r_0(\nu))}\nonumber\\
\mathrm{s.t.}~&~(\ref{eqn:P2:con1}),~(\ref{eqn:P2:con2}),~{\rm and}~(\ref{eqn:Pave}).\nonumber
\end{align}
By introducing an auxiliary variable $t$, problem (P3.1) is equivalently expressed as
\begin{align}
{\rm (P3.2)}:\max_{\{q(\nu)\},t}&~t\nonumber\\
\mathrm{s.t.}~&~{\mathbb{E}_\nu(\hat r_0(\nu)X(\nu))} \ge t {\mathbb{E}_\nu(\hat r_0(\nu))}\label{problem:p3:1:con1}\\
&~(\ref{eqn:P2:con1}),~(\ref{eqn:P2:con2}),~{\rm and}~(\ref{eqn:Pave}).\nonumber
\end{align}
To solve problem (P3.2), we use a similar approach as for solving (P2.1) in Section \ref{sec:III}, in which we first solve the following feasibility problem under any given $t$ and then search $t$ via bisection over the regime $[0,1]$.
\begin{align}
{\rm (P3.3)}:\mathrm{find}~&~\{q(\nu)\}\nonumber\\
\mathrm{s.t.}~&~(\ref{eqn:P2:con1}),~(\ref{eqn:P2:con2}),~(\ref{eqn:Pave}),~{\rm and}~(\ref{problem:p3:1:con1}).\nonumber
\end{align}
In the rest of this subsection, we focus on solving the feasibility problem (P3.3).

First, it can be verified similarly as for Lemmas \ref{lemma:1} and  \ref{lemma:2} that strong duality holds between problem (P3.3) and its dual problem. As a result, we use the Lagrange duality method to solve this problem. Let the dual variables associated with the constraints in (\ref{problem:p3:1:con1}), (\ref{eqn:P2:con1}), and (\ref{eqn:Pave}) be denoted by $\mu \ge 0$, $\lambda \ge 0$, $\zeta$, respectively. Then the partial Lagrangian of problem (P3.3) is denoted as
\begin{align}
\mathcal{L}_3&(\{q(\nu)\},\mu,\lambda,\zeta)
 = \mu\left({\mathbb{E}_\nu((X(\nu)-t)\hat r_0(\nu))}\right) \nonumber\\ &- \lambda \left(\mathbb{E}_{\nu}(q(\nu)) - Q\right) - \zeta\left(\mathbb{E}_\nu\left(\hat p(\nu)\right) - P\right).
\end{align}
As a result, the dual function of (P3.3) is expressed as
\begin{align}
f_3(\mu,\lambda,\zeta)=\max_{\{q(\nu)\ge 0\}}~\mathcal{L}_3(\{q(\nu)\},\mu,\lambda,\zeta).\label{eqn:P2:2:dual:function}
\end{align}
The dual problem is accordingly written as
\begin{align}
\min_{\mu\ge0,\lambda\ge 0,\zeta}~f_3(\mu,\lambda,\zeta).\label{eqn:P2:2:dual:problem}
\end{align}
Here, the optimal value of the dual problem (\ref{eqn:P2:2:dual:problem}) is zero when problem (P3.3) is feasible, while it approaches to $-\infty$ otherwise.

Similar to Proposition \ref{proposition:P2:1}, we have the following proposition, whose proof is omitted for brevity. This proposition can be used for checking the feasibility of problem (P3.3) later.
\begin{proposition}\label{proposition:6}
Problem (\ref{eqn:29:feasibility}) is infeasible if and only if there exist $\mu\ge 0$, $\lambda \ge 0$ and $\zeta$ such that $f_3(\mu, \lambda,\zeta) < 0$.
\end{proposition}

Furthermore, we have the following proposition to find the optimal solution to problem (\ref{eqn:P2:2:dual:function}) to obtain $f_3(\mu,\lambda,\zeta)$ under given any $\mu\ge 0$, $\lambda\ge 0$, and $\zeta$.
\begin{proposition}\label{proposition:5}
The optimal solution to problem (\ref{eqn:P2:2:dual:function}) is given as
\begin{align}
&q_3^{(\mu,\lambda,\zeta)}(\nu) = \nonumber\\
&\left\{
\begin{array}{ll}
{\hat q}_1(\nu), &{\rm if}~\hat v_1(\nu) > \hat v_3(\nu)~\\&~~~{\rm and}~\left(\hat v_1(\nu) > \hat v_2(\nu)~{\rm or}~{\hat q}_2(\nu) \ge {\hat q}_1(\nu)\right)\\
{\hat q}_2(\nu), &{\rm if}~\hat v_2(\nu) > \hat v_1(\nu)~\\&~~~{\rm and}~\hat v_2(\nu) > \hat v_3(\nu)~{\rm and}~{\hat q}_2(\nu) < {\hat q}_1(\nu)\\
{\hat q}_3(\nu), &{\rm if}~\hat v_3(\nu) > \hat v_1(\nu)~\\&~~~{\rm and}~\left(\hat v_3(\nu) > \hat v_2(\nu)~{\rm or}~{\hat q}_2(\nu) \ge {\hat q}_1(\nu)\right)\\
\end{array}
\right.,\forall \nu.
\end{align}
Here,
\begin{align}\label{eqn:hatq1}
\hat q_1(\nu) \triangleq \left[\frac{g_0(\nu)}{\ln2\cdot \beta g_2(\nu)} - \frac{\sigma_0^2}{g_2(\nu)}\right]^+
\end{align}
denotes the jamming power such that the suspicious transmitter does not allocate any power over the fading state $\nu$ (due to the water-filling power allocation),
\begin{align}\label{eqn:hatq2}
\hat q_2(\nu) \triangleq \left[\left(\frac{g_0(\nu)}{g_1(\nu)}\sigma^2_1 - \sigma^2_0\right)\frac{1}{g_2(\nu)}\right]^+
\end{align}
means the minimum jamming power for the legitimate monitor to successfully eavesdrop the suspicious communication, and
\begin{align}\label{eqn:hatq3}
{\hat q}_3(\nu) = \left[\min(\hat q_1(\nu),\hat q_2(\nu),\hat q_4(\nu))\right]^+
\end{align}
represents the used jamming power when the legitimate monitor cannot successfully eavesdrop with
\begin{align}
\hat q_4(\nu) \triangleq \frac{g_0(\nu)t\mu}{\ln2\cdot(g_0(\nu)\lambda - \zeta g_2(\nu))} - \frac{\sigma_0^2}{g_2(\nu)}.\label{eqn:q4}
\end{align}
Accordingly, their resultant objective values of problem (\ref{eqn:P2:2:dual:function}) are respectively given by
\begin{align}
\hat v_1(\nu) &= -\lambda\hat q_1(\nu), \label{eqn:hatv1}\\
\hat v_2(\nu) &= \mu(1-t)\log_2\left( \frac{g_0(\nu)}{\ln 2 \cdot \beta(g_2(\nu)\hat q_2(\nu) + \sigma^2_0)} \right) \nonumber\\ &- \lambda \hat q_2(\nu) - \zeta \left(\frac{1}{\ln 2 \cdot \beta} - \frac{g_2(\nu)\hat q_2(\nu)+\sigma^2_0}{g_0(\nu)} \right), \label{eqn:hatv2} \\
\hat v_3(\nu) &= -t\mu\log_2\left( \frac{g_0(\nu)}{\ln 2 \cdot \beta(g_2(\nu)\hat q_3(\nu)  + \sigma^2_0)} \right)  \nonumber\\ & - \lambda \hat q_3(\nu)  - \zeta \left(\frac{1}{\ln 2 \cdot \beta} - \frac{g_2(\nu)\hat q_3(\nu) +\sigma^2_0}{g_0(\nu)} \right). \label{eqn:hatv3}
\end{align}
\end{proposition}
\begin{IEEEproof}
See Appendix \ref{app:E}.
\end{IEEEproof}
With the optimal solution to problem (\ref{eqn:P2:2:dual:function}) given in Proposition \ref{proposition:5} together with Proposition \ref{proposition:6}, we can then apply the ellipsoid method to solve the dual problem (\ref{eqn:P2:2:dual:problem}) by using the fact that the subgradient of $f_3(\mu,\lambda,\zeta)$ is
\begin{align*}
\mv s_3(\mu,\lambda,\zeta) = &\left[\mathbb{E}_\nu\left(\left(X^{(\mu,\lambda,\zeta)}(\nu) - t\right) \hat{r}_0^{(\mu,\lambda,\zeta)}(\nu)\right),\right. \nonumber\\
&\left.Q - {\mathbb{E}_\nu(q_3^{(\mu,\lambda,\zeta)}(\nu))}, P - \mathbb{E}_\nu\left(\hat p^{(\mu,\lambda,\zeta)}(\nu)\right)\right]^T.
\end{align*}
Here, $\{\hat{r}_0^{(\mu,\lambda,\zeta)}(\nu)\}$, $\{X^{(\mu,\lambda,\zeta)}(\nu)\}$, and $\{\hat p^{(\mu,\lambda,\zeta)}(\nu)\}$ denote the corresponding $\{\hat r_0(\nu)\}$ in (\ref{eqn:hatr}), $\{X(\nu)\}$ in (\ref{eqn:indictor}), and $\{\hat p(\nu)\}$ in (\ref{eqn:solution:p}) under given $\{q_3^{(\mu,\lambda,\zeta)}(\nu)\}$, respectively.

Note that the detailed complete algorithm for solving problem (P3.3) is similar to Algorithm 1, and thus is omitted for brevity. Therefore, problem (P3.3) is finally solved. When problem (P3.3) is feasible, we denote its optimal solution as $\{q_3^{\star}(\nu)\}$.

Finally, we use the bisection search to find the optimal $t$ for problem (P3.2) under any given $\beta \in [\beta^{\min},\beta^{\max}]$, and apply the exhaustive search to obtain the optimal $\beta$ for problem (P3). Let the optimal $\beta$ for problem (P3) and the correspondingly optimal $t$ for problem (P3.2) be denoted by $\beta^{**}$ and $t^{**}$, respectively. As a result, the accordingly obtained $\{q_3^{\star}(\nu)\}$ becomes the optimal cognitive jamming power solution to (P3), denoted by $\{q_3^{*}(\nu)\}$. Therefore, problem (P3) is solved.

\begin{figure}
\centering
 \epsfxsize=1\linewidth
    \includegraphics[width=7cm]{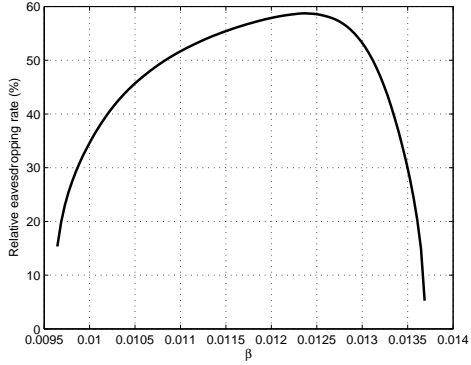}
\caption{The relative eavesdropping rate versus the variable $\beta$ when the suspicious transmitter adopts the water-filling power allocation in the delay-tolerant case.} \label{fig:5}\vspace{-1em}
\end{figure}

\begin{remark}
To provide more insight, Fig. \ref{fig:5} shows the obtained relative eavesdropping rate (the optimal value of (P3.1)) under given $\beta$ versus the variable $\beta$ (in the range between $\beta^{\min}$ and $\beta^{\max}$), where the system parameters are set as in Section \ref{sec:VI} and the average jamming power at the legitimate monitor is set to be $Q = 20$ dB. It is observed that the relative eavesdropping rate first increases and then decreases as a function of $\beta$. Note that we have also conducted simulations under other setups which are not plotted here, and such a property is also shown to be valid under these tests, although it is very difficult to rigorously prove it. This property implies that a simple bisection (instead of the complex exhaustive search) may be sufficient to find the optimal $\beta$ for problem (P3). As a result, the complexity of solving problem (P3) can be significantly reduced.
\end{remark}

%\begin{table}[!t]\scriptsize
%%\renewcommand{\arraystretch}{1.3}
%\caption{Complete Algorithm for Solving Problem (P3)} \centering
%\begin{tabular}{|p{8cm}|}
%\hline
%\textbf{Algorithm 2}\\
%\hline\vspace{0.01cm}
%1) {\bf Initialization}: Obtain the upper and lower bounds of $\beta$ as $\beta^{\max}$ and $\beta^{\min}$ as in Section \ref{sec:VA}, and generate the candidate set of $\beta$ as $\{\beta^{\min}, \beta^{\min}+\epsilon, \beta^{\min}+2\epsilon,\ldots,\beta^{\max}\}$, where $\epsilon > 0$ denotes the step size for the exhaustive search over $\beta$. Set $\bar \beta = \beta^{\min}$ and $\bar t = 0$.\\
%2) {\bf Repeat:}
%  \begin{itemize} \setlength{\itemsep}{0pt}
%    \item[a)] Solve problem (P3.2) under given $\beta = \bar\beta$ by using bisection search over $t$, and solving problem (P3.3) via the Lagrange duality method. Let the optimal solution be denoted as $t^{(\bar\beta)}$ and $\{p^{(\bar\beta)}(\nu)\}$.
%    \item[b)] If $t^{\bar\beta} > \bar t$, then set $\beta^* = \bar\beta$ and $p^*(\nu) = p^{(\bar\beta)}(\nu), \forall \nu$.
%    \item[d)] $\bar \beta \gets \bar \beta  + \epsilon$;
%  \end{itemize}
%  3) {\bf Until} $\bar \beta > \beta^{\max}$.\\
% \hline
%\end{tabular}\label{table:2}\vspace{-0em}
%\end{table}

\subsection{Comparison Among Different Optimal Cognitive Jamming Solutions}

\begin{figure*}
\begin{minipage}[t]{0.32\linewidth}
\centering
\includegraphics[width=\textwidth]{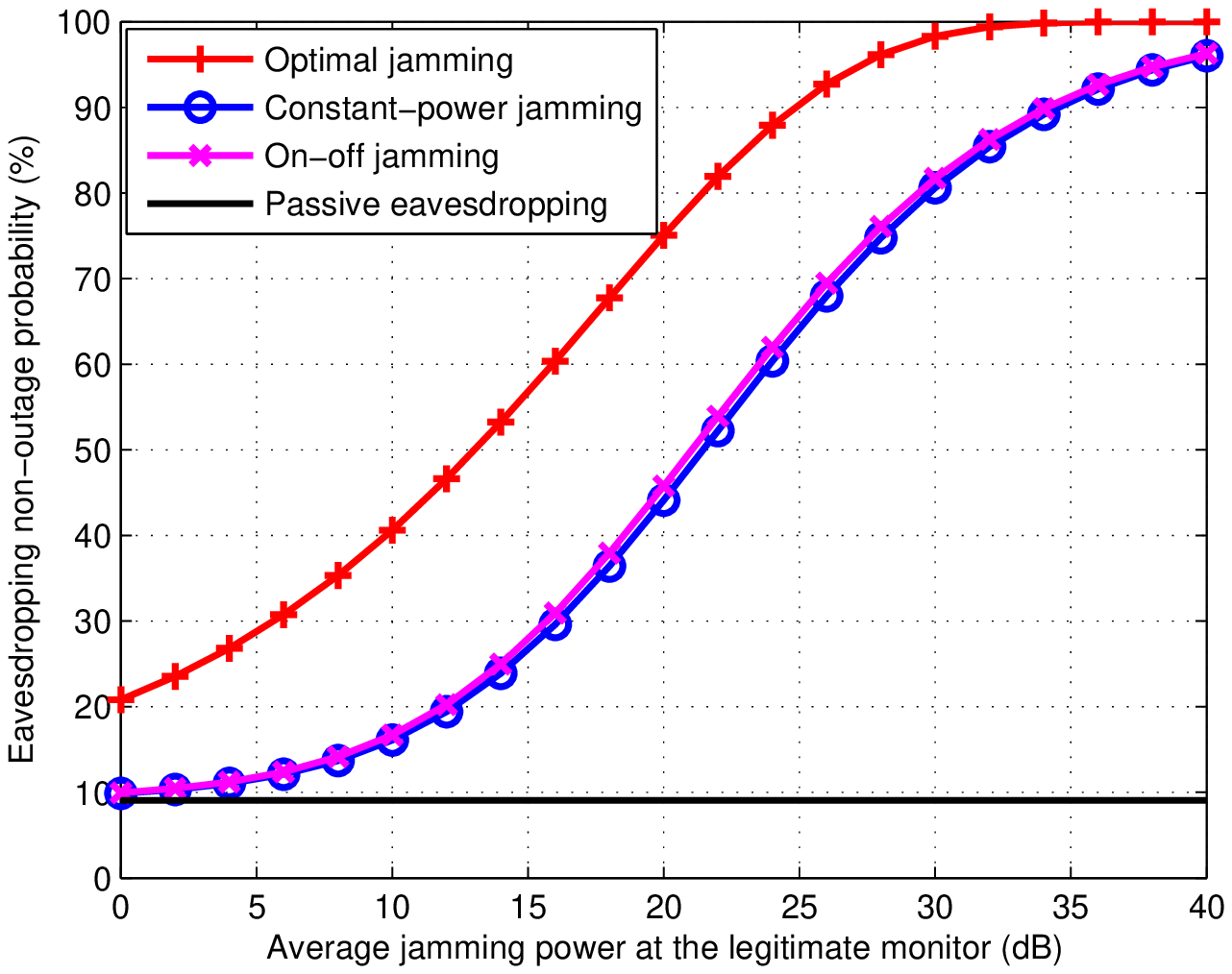}
\caption{The eavesdropping non-outage probability $\mathbb{E}_{\nu}(X(\nu))$ versus the average jamming power $Q$ at the legitimate monitor in delay-sensitive suspicious applications.}\label{fig:2}
\end{minipage}
\hfill
\begin{minipage}[t]{0.32\linewidth}
\centering
\includegraphics[width=\textwidth]{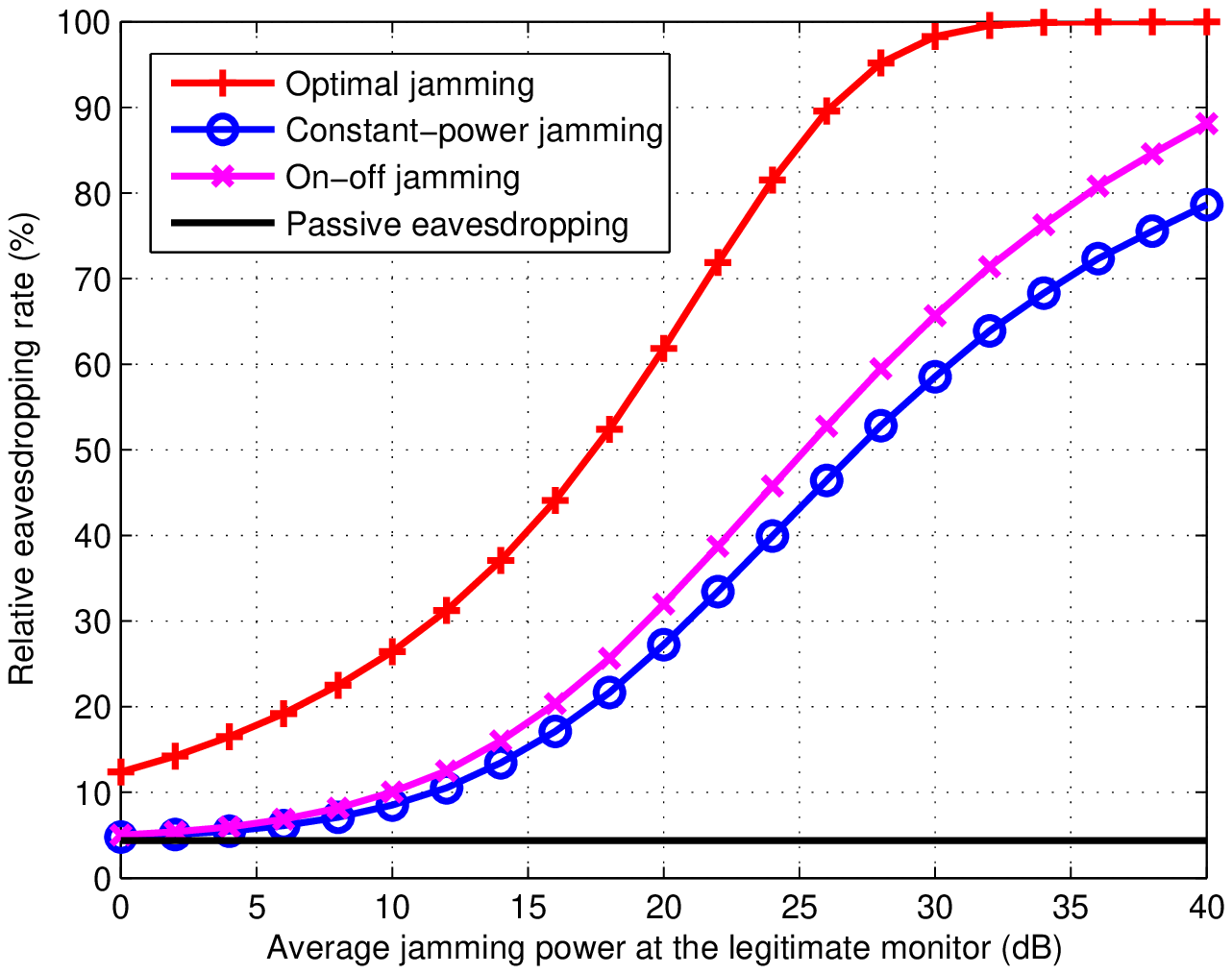}
\caption{The relative eavesdropping rate $\frac{\mathbb{E}_\nu(\bar r_0(\nu)X(\nu))}{\mathbb{E}_\nu(\bar r_0(\nu))}$ versus the average jamming power $Q$ at the legitimate monitor in delay-tolerant suspicious applications, where the suspicious transmitter employs the fixed power transmission.}\label{fig:3}
\end{minipage}
\hfill
\begin{minipage}[t]{0.32\linewidth}
\centering
\includegraphics[width=\textwidth]{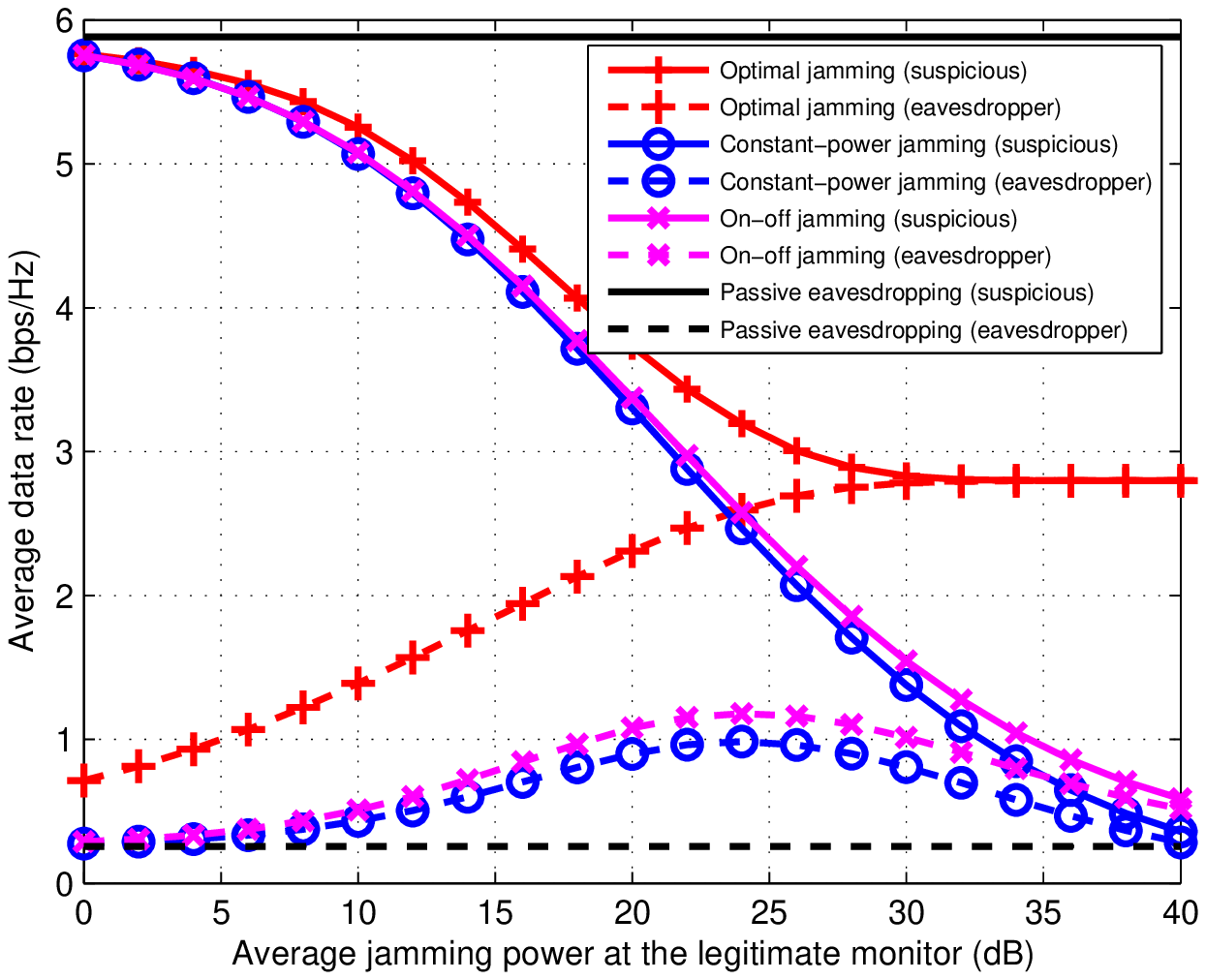}
\caption{The average suspicious communication rate $\mathbb{E}(\bar r_0(\nu))$ and the average eavesdropping rate $\mathbb{E}(\bar r_0(\nu)X(\nu))$) versus the average jamming power $Q$ at the legitimate monitor in delay-tolerant suspicious applications, where the suspicious transmitter employs the fixed power transmission.}\label{fig:4}
\end{minipage}\vspace{-1em}
\end{figure*}

In this subsection, we compare the optimal cognitive jamming solutions under different application scenarios, i.e., $\left\{q_1^*(\nu)\right\}$, $\left\{q_2^*(\nu)\right\}$, and $\left\{q_3^*(\nu)\right\}$ for problems (P1), (P2), and (P3), respectively.

First, consider each fading state $\nu$ with $\left(\frac{g_0(\nu)}{g_1(\nu)}\sigma^2_1 - \sigma^2_0\right)\frac{1}{g_2(\nu)} \le 0$, where the eavesdropping link is better than the suspicious communication link. In this case, since eavesdropping is always successful and no jamming is required, the three optimal solutions are identical, i.e., $q_3^*(\nu) = q_2^*(\nu) = q_1^*(\nu) = 0$.

Next, consider each of the other fading states $\nu$ with $\left(\frac{g_0(\nu)}{g_1(\nu)}\sigma^2_1 - \sigma^2_0\right)\frac{1}{g_2(\nu)} > 0$, where the eavesdropping link is worse than the suspicious communication link and thus the eavesdropping will not be successful without proactive jamming. In this case, the optimal solution $\{q_1^*(\nu)\}$  for delay-sensitive applications are different from $\{q_2^*(\nu)\}$ and $\{q_3^*(\nu)\}$ for delay-tolerant applications. Specifically, in delay-sensitive applications, the legitimate monitor only jams over the desired fading states when it can successfully eavesdrop (after jamming); while in the delay-tolerant case, the legitimate monitor may also jam over the undesired fading states when it cannot successfully eavesdrop (even after jamming), since it helps reduce the communication rate of the suspicious link in these fading states and therefore increase the percentage of successful eavesdropping rate in the desired channel states.

Finally, note that the optimal solutions $\{q_2^*(\nu)\}$ and $\{q_3^*(\nu)\}$ are also different from each other at each fading state $\nu$ with $\left(\frac{g_0(\nu)}{g_1(\nu)}\sigma^2_1 - \sigma^2_0\right)\frac{1}{g_2(\nu)} > 0$. For example, when the suspicious transmitter can adjust its transmit power via water-filling, the legitimate monitor may choose between the jamming power $\hat q_1(\nu)$ such that the suspicious transmitter does not allocate power over the fading state $\nu$, versus $\hat q_2(\nu)$ such that the legitimate monitor can eavesdrop the suspicious link over that fading state. In contrast, when the suspicious transmitter adopts fixed power transmission, the legitimate monitor does not need to consider the first option. It is also worth noting that it is difficult for us to analytically compare the resulting relative eavesdropping rate under fixed power transmission with that under water-filling power allocation. As will be shown in the numerical results later (see Fig. \ref{fig:8} in Section \ref{sec:VI}), the relative eavesdropping rate under fixed power transmission is higher than that under water-filling power allocation. This implies that due to the potential water-filling power allocation at the suspicious transmitter, in general higher average jamming power is required for the legitimate monitor to achieve the same relative eavesdropping rate as in the case with fixed power transmission.

\section{Numerical Results}\label{sec:VI}\vspace{-0em}

%
%\begin{figure}
%\centering
% \epsfxsize=1\linewidth
%    \includegraphics[width=10cm]{Fig_Probability.eps}
%\caption{The eavesdropping non-outage probability $\mathbb{E}_{\nu}(X(\nu))$ versus the average jamming power $Q$ at the legitimate monitor in delay-sensitive suspicious applications.} \label{fig:2}\vspace{-0em}
%\end{figure}
In this section, we provide numerical results to validate the performance of our proposed proactive eavesdropping via cognitive jamming approach, in terms of the eavesdropping non-outage probability $\mathbb{E}_{\nu}(X(\nu))$ and the relative eavesdropping rate $\frac{\mathbb{E}_\nu(\bar r_0(\nu)X(\nu))}{\mathbb{E}_\nu(\bar r_0(\nu))}$. For comparison, we consider three benchmark schemes as follows: 1) {\it Proactive eavesdropping with constant-power jamming}: in this scheme, the legitimate monitor uses constant jamming power over all fading states, i.e., $q(\nu) = Q, \forall \nu$. 2) {\it Proactive eavesdropping with ``on-off'' jamming}: in this scheme, the legitimate monitor does not send any jamming signal over the fading state $\nu$ with $\frac{g_0(\nu)}{\sigma_0^2} \le \frac{g_1(\nu)}{\sigma_1^2}$ (i.e., the eavesdropping is successful even without any jamming), and allocates the jamming power equally over all the other fading states. 3) {\it Passive eavesdropping without jamming}: in this scheme, the legitimate monitor does not send any jamming signal, i.e., $q(\nu) = 0, \forall \nu$.

In the simulation, we consider Rayleigh fading and set the channel coefficients $h_0(\nu)$, $h_1(\nu)$, and $h_2(\nu)$ to be independent CSCG random variables with mean zero and variances $1,0.1$, and $0.1$, respectively, $\forall \nu$, by assuming that the legitimate monitor is far away from the suspicious transmitter and receiver as compared to their distance. Furthermore, we set the transmit power at the suspicious transmitter to be $P= 20$ dB unless otherwise stated,  and the noise powers to be $\sigma_0^2 = \sigma_1^2 = 1$. Here, the system parameters are normalized without loss of generality, and can be easily extended to the case with realistic parameters.

First, consider delay-sensitive suspicious applications. Fig. \ref{fig:2} shows the eavesdropping non-outage probability $\mathbb{E}_{\nu}(X(\nu))$ versus the average jamming power $Q$ at the legitimate monitor. It is observed that the proactive eavesdropping (with both cognitive jamming and constant-power jamming) achieves higher eavesdropping non-outage probability than the passive eavesdropping, while the cognitive jamming with optimal power control outperforms the constant-power jamming.

\begin{figure*}
\begin{minipage}[t]{0.32\linewidth}
\centering
\includegraphics[width=\textwidth]{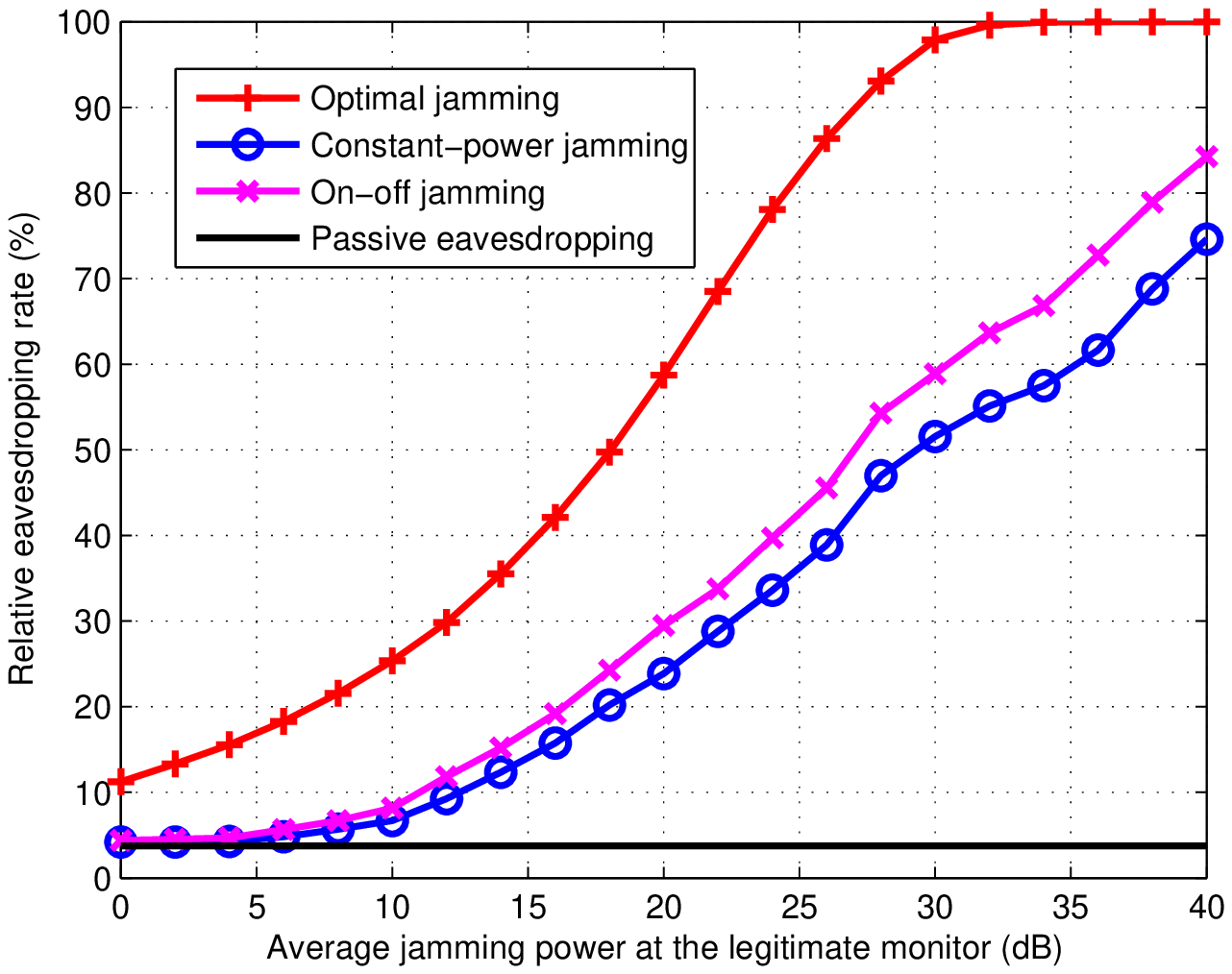}
\caption{The relative eavesdropping rate $\frac{\mathbb{E}_\nu(\hat r_0(\nu)X(\nu))}{\mathbb{E}_\nu(\hat r_0(\nu))}$ versus the average jamming power $Q$ at the legitimate monitor in delay-tolerant suspicious applications, where the suspicious transmitter employs the water-filling power allocation.}\label{fig:6}
\end{minipage}
\hfill
\begin{minipage}[t]{0.32\linewidth}
\centering
\includegraphics[width=\textwidth]{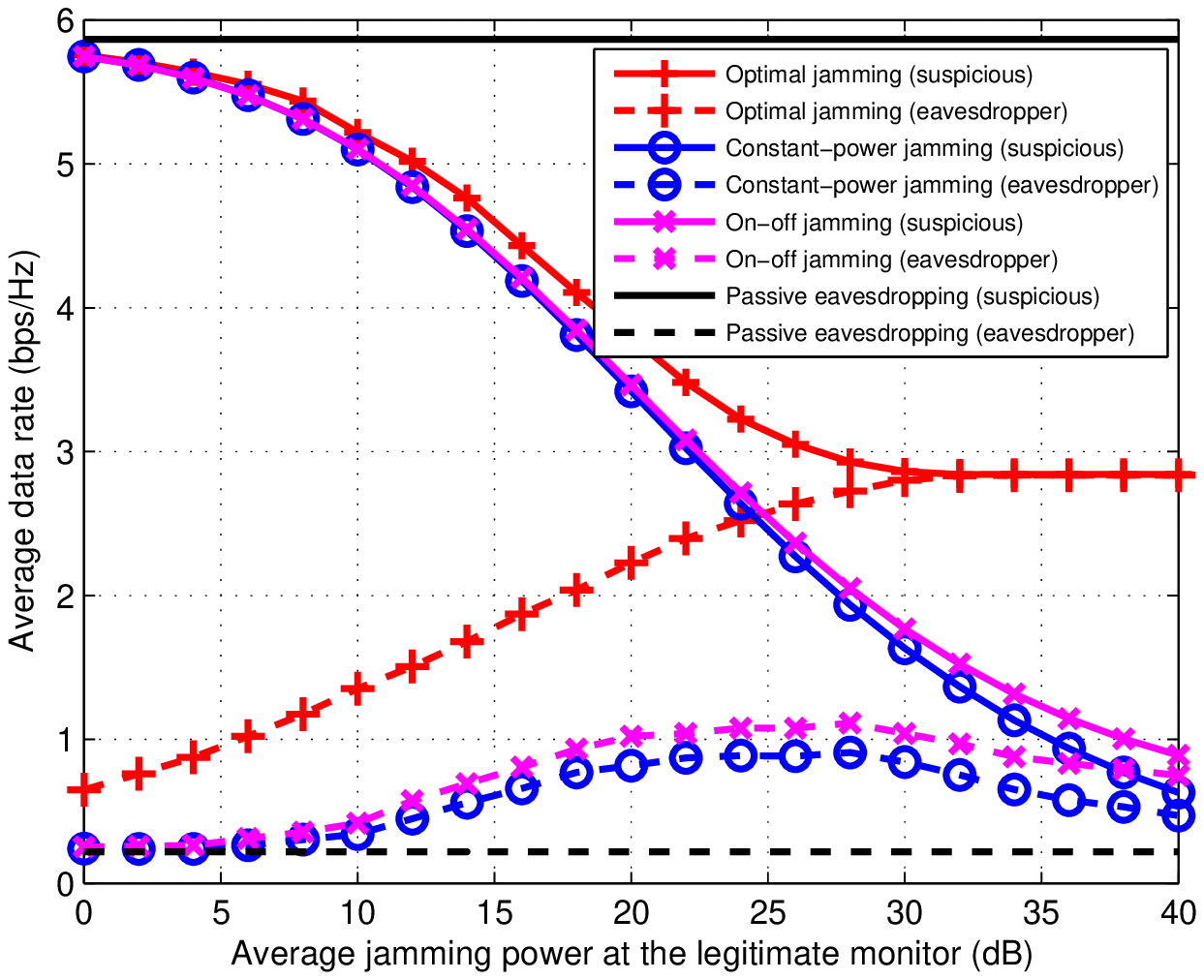}
\caption{The average suspicious communication rate $\mathbb{E}(\hat r_0(\nu))$ and the average eavesdropping rate $\mathbb{E}(\hat r_0(\nu)X(\nu))$) versus the average jamming power $Q$ at the legitimate monitor in delay-tolerant suspicious applications, where the suspicious transmitter employs the water-filling power allocation.}\label{fig:7}
\end{minipage}
\hfill
\begin{minipage}[t]{0.32\linewidth}
\centering
\includegraphics[width=\textwidth]{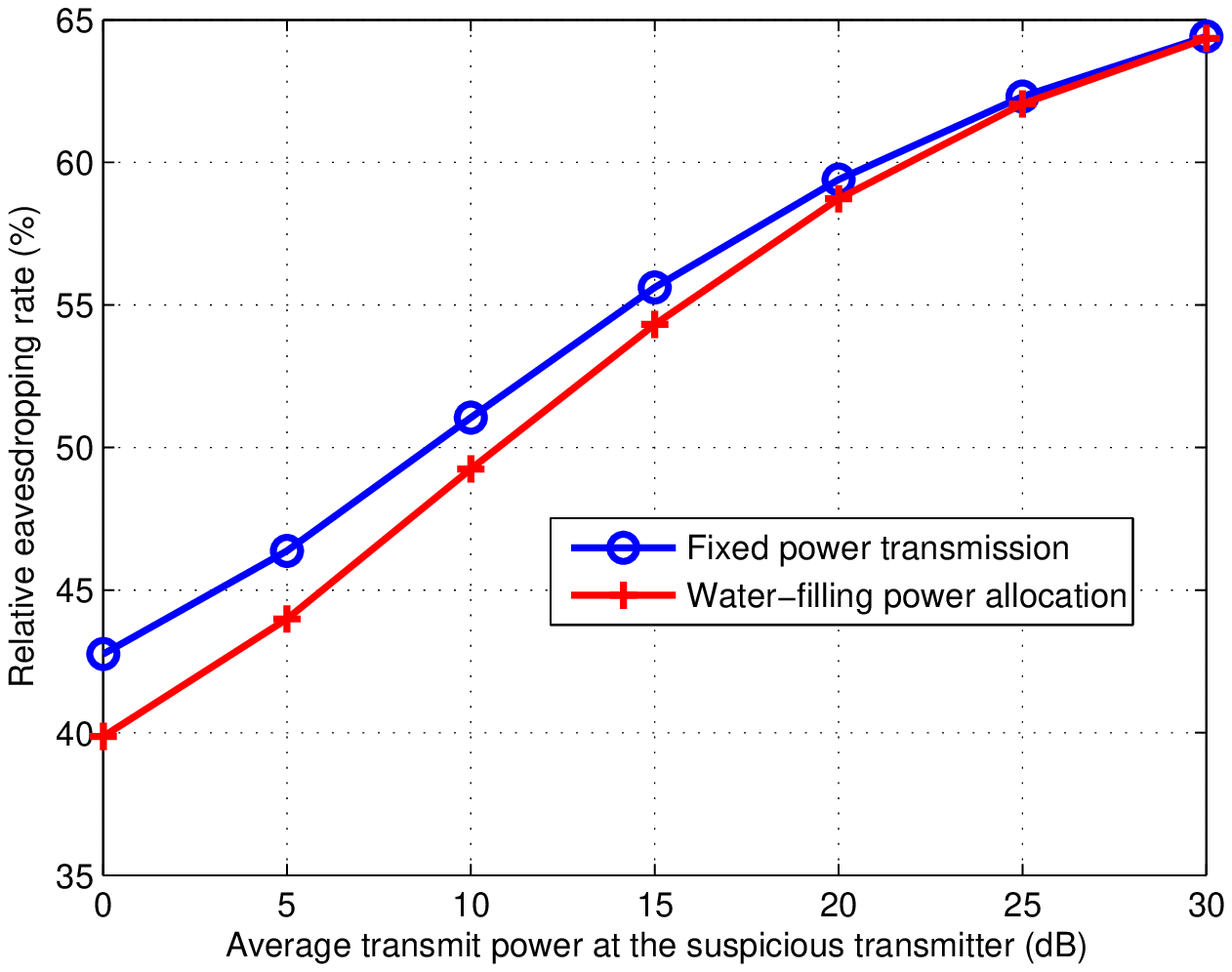}
\caption{The relative eavesdropping rates versus the average transmit power $P$ at the suspicious transmitter in delay-tolerant suspicious applications.}\label{fig:8}
\end{minipage}\vspace{-1em}
\end{figure*}
%
%\begin{figure}
%\centering
% \epsfxsize=1\linewidth
%    \includegraphics[width=8cm]{Fig1_waterfilling.eps}
%\caption{The relative eavesdropping rate $\frac{\mathbb{E}_\nu(\hat r_0(\nu)X(\nu))}{\mathbb{E}_\nu(\hat r_0(\nu))}$ versus the average jamming power $Q$ at the legitimate monitor in delay-tolerant suspicious applications, where the suspicious transmitter employs the water-filling power allocation.}\label{fig:6}
%\end{figure}

Next, consider that the suspicious transmitter employs fixed power transmission in delay-tolerant suspicious applications. Fig. \ref{fig:3} shows the relative eavesdropping rate $\frac{\mathbb{E}_\nu(\bar r_0(\nu)X(\nu))}{\mathbb{E}_\nu(\bar r_0(\nu))}$ versus the average jamming power $Q$ at the legitimate monitor. The proactive eavesdropping via cognitive jamming achieves the best eavesdropping performance in terms of the relative eavesdropping rate. Furthermore, Fig. \ref{fig:4} shows the average suspicious communication rate $\mathbb{E}(\bar r_0(\nu))$ and the average eavesdropping rate $\mathbb{E}(\bar r_0(\nu) X(\nu))$, respectively. It is observed that as compared to the constant-power jamming, the cognitive jamming with optimal power control achieves higher average eavesdropping rate. This shows that by utilizing the optimal cognitive jamming, the legitimate monitor can not only eavesdrop a higher percentage of data bits but also a larger volume of data. This validates the advantages of the proposed proactive eavesdropping via cognitive jamming with the optimal power control.
%
%\begin{figure}
%\centering
% \epsfxsize=1\linewidth
%    \includegraphics[width=8cm]{Fig2_waterfilling.eps}
%\caption{The average suspicious communication rate $\mathbb{E}(\hat r_0(\nu))$ and the average eavesdropping rate $\mathbb{E}(\hat r_0(\nu)X(\nu))$) versus the average jamming power $Q$ at the legitimate monitor in delay-tolerant suspicious applications, where the suspicious transmitter employs the water-filling power allocation.}\label{fig:7}
%\end{figure}
%
%\begin{figure}
%\centering
% \epsfxsize=1\linewidth
%    \includegraphics[width=8cm]{Fig7.eps}
%\caption{The relative eavesdropping rates versus the average transmit power $P$ at the suspicious transmitter in delay-tolerant suspicious applications.}\label{fig:8}
%\end{figure}

In addition, consider the suspicious transmitter employs water-filling power allocation in delay-tolerant suspicious applications. Fig. \ref{fig:6} shows the relative eavesdropping rate $\frac{\mathbb{E}_\nu(\hat r_0(\nu)X(\nu))}{\mathbb{E}_\nu(\hat r_0(\nu))}$ versus the average jamming power $Q$ at the legitimate monitor, and Fig. \ref{fig:7} shows the average suspicious communication rate $\mathbb{E}(\hat r_0(\nu))$ and the average eavesdropping rate $\mathbb{E}(\hat r_0(\nu)X(\nu))$, respectively. The two figures can be similarly explained as for Figs. \ref{fig:3} and \ref{fig:4}, respectively.

Finally, Fig. \ref{fig:8} shows the relative eavesdropping rates versus the average transmit power $P$ at the suspicious transmitter in delay-tolerant suspicious applications, where both fixed power transmission and water-filling power allocation at the suspicious transmitter are considered. Here, the average jamming power at the legitimate monitor is set to be $Q = 20$ dB. It is observed that under the same average jamming power, the relative eavesdropping rate under the fixed power transmission at the suspicious transmitter is higher than that under the water-filling power allocation, especially when the average transmit power $P$ at the suspicious transmitter is small. This shows that the dynamics of water-filling power allocation at the suspicious transmitter degrades the proactive eavesdropping performance at the legitimate monitor, and the legitimate monitor needs to use higher average jamming power when the suspicious transmitter adopts water-filling power allocation to achieve the same performance as that under fixed power transmission.

\section{Practical Implementation of Proactive Eavesdropping}\label{sec:Pro:Eav}

Preceding sections focused on characterizing the fundamental information-theoretical limits of proactive eavesdropping under the assumption with perfect SIC and global CSI at the legitimate monitor. In this section, we consider a more practical case with residual SI and local CSI only, and accordingly design an efficient {\it online} cognitive jamming scheme, inspired by the optimal cognitive jamming above. In the following, we particularly focus on the eavesdropping non-outage probability maximization problem for delay-sensitive applications. Similar ideas and analysis can be used to address the relative eavesdropping rate maximization problems for delay-tolerant applications, but the details are omitted here due to space limitation.

\subsection{Eavesdropping Non-Outage Probability Maximization in the Case with Residual SI}

First, we investigate the effect of the residual SI at the legitimate monitor by assuming the loop-back channel power gain from the jamming to the eavesdropping antennas as $\tilde\phi(\nu)$ in the fading state $\nu$. Suppose that the SIC at the legitimate monitor achieves an SI reduction of $\varphi$ (in dB). Then, the residual SI in this fading state is given as $\tilde\phi(\nu)q(\nu)/\varphi = \phi(\nu)q(\nu)$, where $\phi(\nu) = \tilde\phi(\nu)/\varphi$ denotes the effective loop-back channel power gain after SIC. As demonstrated in practical full-duplex radios \cite{fullduplex}, jointly using analog and digital SIC methods can achieve up to 110 dB SI reduction. With the residual SI, the SNR at the legitimate monitor in (\ref{eqn:gamma1}) can be revised as the following SINR:
\begin{align}
\tilde\gamma_1(\nu) = \frac{g_1(\nu)p(\nu)}{\phi(\nu) q(\nu) + \sigma^2_1}.\label{eqn:hat:gamma_1}
\end{align}
In this case, the successful eavesdropping indicator function in (\ref{eqn:indictor}) is rewritten as
\begin{align}
\tilde X(\nu) = \left\{
\begin{array}{ll}
1, & {\rm if}~\tilde\gamma_1(\nu) \ge \gamma_0(\nu)\\
0, & {\rm otherwise}.\\
\end{array}
\right.
\end{align}
The non-outage eavesdropping non-outage probability maximization problem (P1) is thus re-expressed as
\begin{align*}
\mathrm{(P4)}:\max_{\{q(\nu) \ge 0\}}~&~ \mathbb{E}_{\nu}\left(\tilde X(\nu)\right)\\
\mathrm{s.t.}~&~(\ref{eqn:P2:con1}).
\end{align*}
We have the following proposition.
\begin{proposition}\label{proposition:P4}
The optimal solution to problem (P4) is given as
\begin{align}
&q_4^{*}(\nu) = \nonumber\\
&\left\{
\begin{array}{ll}
 \frac{g_0(\nu) \sigma^2_1 - g_1(\nu) \sigma^2_0}{g_1(\nu)g_2(\nu) - g_0(\nu)\phi(\nu)},& {\rm if}~g_1(\nu)g_2(\nu) - g_0(\nu)\phi(\nu) > 0~ \\&{\rm and}~0 <  \frac{g_0(\nu) \sigma^2_1 - g_1(\nu) \sigma^2_0}{g_1(\nu)g_2(\nu) - g_0(\nu)\phi(\nu)} < \frac{1}{\lambda_4^*}\\
0,&{\rm otherwise},
\end{array}
\right.
\end{align}
where $\lambda_4^*$ denotes the optimal dual variable associated with the constraint (\ref{eqn:P2:con1}).
\end{proposition}
\begin{proof}[Sketch of Proof]
Note that strong duality still holds between (P4) and its Lagrange dual problem. Therefore, this proposition can be verified by applying the Lagrange duality method to solve (P4), similarly as in Proposition \ref{proposition:P1}. The details are omitted here for brevity.
\end{proof}

Proposition \ref{proposition:P4} shows that at each fading state $\nu$, if the eavesdropping link is weaker than the suspicious link (i.e., $g_0(\nu) \sigma^2_1 - g_1(\nu) \sigma^2_0 > 0$), then the minimum jamming power for the eavesdropping to be successful is given as $\tilde q^*(\nu)\triangleq \frac{g_0(\nu) \sigma^2_1 - g_1(\nu) \sigma^2_0}{g_1(\nu)g_2(\nu) - g_0(\nu)\phi(\nu)}$, which is valid only when the residual SI is not so strong (i.e., $g_1(\nu)g_2(\nu) - g_0(\nu)\phi(\nu) > 0$ holds). By comparing Proposition \ref{proposition:P1} versus Proposition \ref{proposition:P4}, we observe that threshold-based jamming power allocations are optimal to maximize the eavesdropping non-outage probability in both cases without and with residual SI, where the jamming power cannot exceed the thresholds $\frac{1}{\lambda_1^*}$ and $\frac{1}{\lambda_4^*}$, respectively.

\subsection{Online Cognitive Jamming Under Practical Assumptions}

Inspired by the optimal threshold-based power allocation in Proposition \ref{proposition:P4}, we then consider online cognitive jamming strategies under the following practical assumptions. First, instead of considering the case with infinite fading states, we consider a finite horizon of $N$ time blocks, with wireless channels being constant over each block. Accordingly, we use $\nu \in \{1,\ldots,N\}$ to denote the index of the time block in this subsection. Next, at each time block, the legitimate monitor does not know the suspicious channel $g_0(\nu)$ or the jamming channel $g_2(\nu)$, but it knows the eavesdropping channel $g_1(\nu)$ and the effective loop-back channel $\phi(\nu)$ via channel estimation based on the received signals. Therefore, it knows the resultant SINR $\tilde\gamma_1(\nu)$ at the itself under any given jamming power. In addition, under any given jamming power, the legitimate monitor can infer the resultant suspicious communication rate $r_0(\nu)$ in (\ref{eqn:r0}) (and accordingly the SINR $\gamma_0(\nu)$ at the legitimate monitor) by analyzing the received signals from the suspicious transmitter.

Under this setup, we propose an online cognitive jamming scheme by separating each time block into two phases: one for learning the required jamming power $\tilde q^*(\nu)$ at that time block, and the other for eavesdropping information. In the following, we first discuss how to learn $\tilde q^*(\nu)$ at the first phase, and then present the design of the thresholds and the corresponding jamming powers over time for the second phase.

\subsubsection{Learning the Required Jamming Power}

At the first phase of each time block, the legitimate monitor estimates the required jamming power $\tilde q^*(\nu) = \frac{g_0(\nu) \sigma^2_1 - g_1(\nu) \sigma^2_0}{g_1(\nu)g_2(\nu) - g_0(\nu)\phi(\nu)}$. At the first glance, this is a very difficult task as it does not know the channels $g_0(\nu)$ and $g_2(\nu)$ at that time block. Fortunately, under any given jamming power employed, the legitimate monitor is able to know the resultant SINRs $\tilde\gamma_1(\nu)$ at the legitimate monitor and $\gamma_0(\nu)$ at the suspicious receiver. As a result, the legitimate monitor knows whether the currently used jamming power is larger or smaller than $\tilde q^*(\nu)$. In this case, by adjusting the jamming power based on a bisection manner, the legitimate monitor is able to find $\tilde q^*(\nu)$ at that time block.

Note that in general, longer learning time results in more accurate estimation of $\tilde q^*(\nu)$ in the first phase, but reduces the length of the second phase for eavesdropping information. Therefore, there exists a tradeoff in designing the length of the two phases to optimize the eavesdropping performance, especially when the wireless channels fluctuate fast (e.g., due to the mobility of suspicious transmitter and receiver) and each time block is with a finite length. In this section, we consider that each time block is sufficiently long and thus the time consumed for estimation in the first phase is negligible.

\subsubsection{Online Threshold-Based Jamming Design}

After $\tilde q^*(\nu)$ is obtained, we propose a practical {\it online} threshold-based cognitive jamming design, inspired by the optimal cognitive jamming solution in Proposition \ref{proposition:P4}. In particular, at each time block $\nu$, the legitimate monitor updates a threshold $\tau(\nu)$ and accordingly obtains the online jamming power as $q_{\rm online}(\nu) = \tilde q^*(\nu)$ when the required jamming power $\tilde q^*(\nu)$ is no larger than the threshold $\tau(\nu)$, and $q_{\rm online}(\nu) = 0$ otherwise. Furthermore, at each time block $\nu$, if the average jamming power so far (i.e., $\frac{1}{\nu}\sum_{i=1}^{\nu} q_{\rm online}(\nu)$) is less than the maximum average power $Q$, we increase $\tau(\nu+1)$ as $\tau(\nu+1) = \tau(\nu) + \chi$ so as to jam over more blocks subsequently; otherwise, we decrease $\tau(\nu+1)$ as $\tau(\nu+1) = \tau(\nu) - \chi$. Here, $\chi > 0$ denotes a constant step size that is a design parameter. To summarize, we list the detailed algorithm in Table \ref{Table:I}.

\begin{table}[!t]\scriptsize
\caption{Algorithm for the Online Threshold-Based Jamming} \centering
\begin{tabular}{|p{8cm}|}
\hline\vspace{0.01cm}
%\textbf{Algorithm 2}\\
%\hline\vspace{0.01cm}
\begin{itemize}
\item {\bf{Initialization}:} set the initial threshold as $\tau(1)$.
\item {\bf For $\nu=1,\ldots,N$}
    \begin{itemize}
    \item {\bf Jamming power design:} if $\tilde q^*(\nu) \le \tau(\nu)$, we have $q_{\rm online}(\nu) = \tilde q^*(\nu)$; otherwise, it follows that $q_{\rm online}(\nu) = 0$;
    \item {\bf Threshold update:}  if $\frac{1}{\nu}\sum_{i=1}^{\nu} q_{\rm online}(\nu)<Q$, we have $\tau(\nu+1) = \tau(\nu) + \chi$; otherwise, $\tau(\nu+1) = \tau(\nu) - \chi$.
    \end{itemize}
\item {\bf End for}
\end{itemize}\\
 \hline
\end{tabular}\label{Table:I}\vspace{0em}
\end{table}

It is worth noting that in the proposed online threshold-based cognitive jamming, the threshold $\tau(\nu)$'s will converge to the optimal threshold $1/\lambda^*_4$ if the step size $\chi$ is sufficiently small and the number of time blocks $N$ is sufficiently large. This is due to the fact that at each time block the value of $\frac{1}{\nu-1}\sum_{i=1}^{\nu-1} p_{\rm online}(\nu) - Q$ can be viewed as a good approximation of the subgradient of the dual problem of (P4), and therefore, the sequence of $1/\tau(\nu)$'s will converge to the optimal dual variable $\lambda^*_4$.

\begin{figure*}
\begin{minipage}[t]{0.325\linewidth}
\centering
\includegraphics[width=\textwidth]{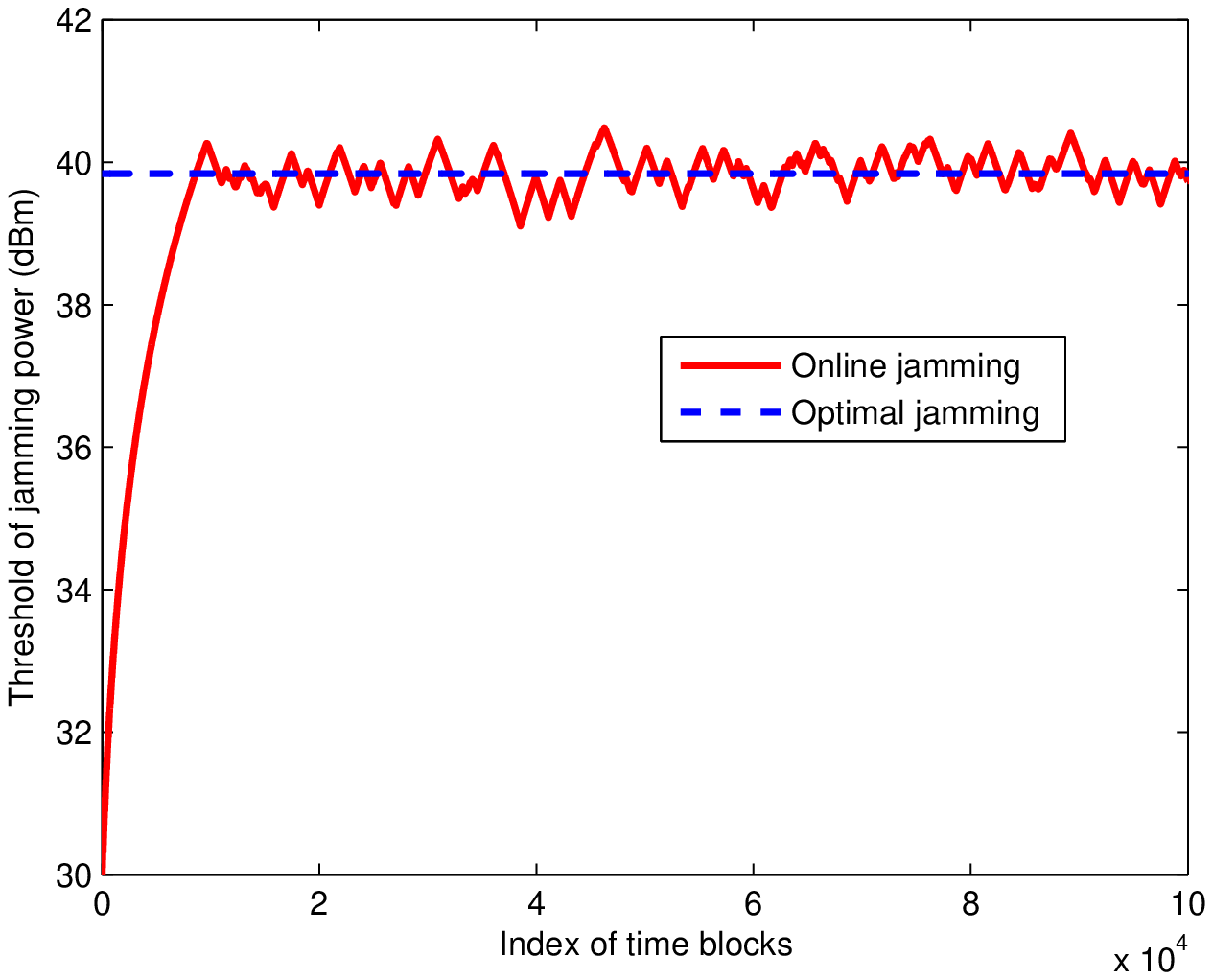}
\caption{The threshold comparison between the optimal and online cognitive jamming, where the maximum average jamming power is set as $Q = 30$ dBm.}\label{fig:9}
\end{minipage}
\hfill
\begin{minipage}[t]{0.325\linewidth}
\centering
\includegraphics[width=\textwidth]{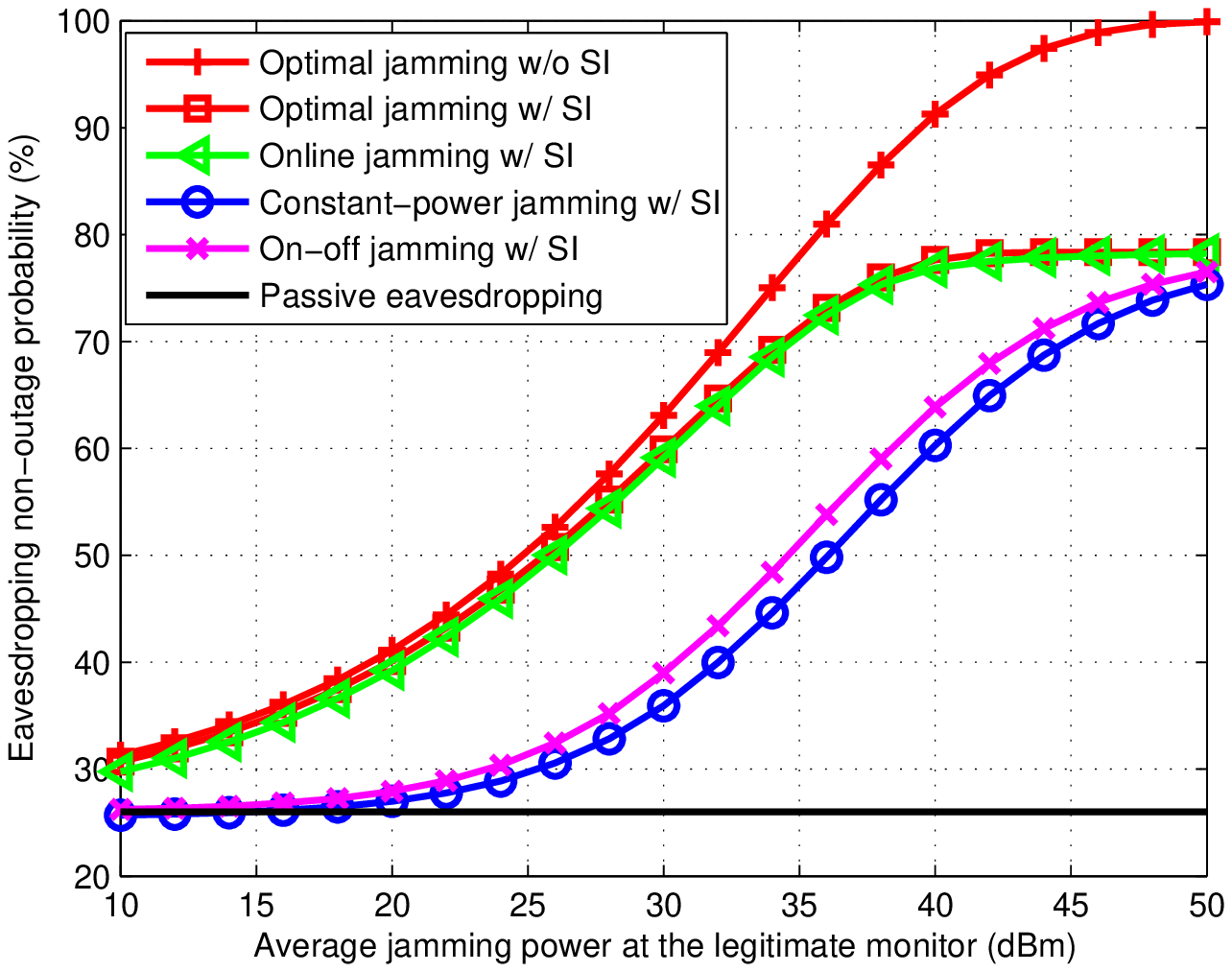}
\caption{Performance comparison between the optimal and online cognitive jamming in the case when the eavesdropping and jamming antennas are co-located at the legitimate monitor.}\label{fig:10}
\end{minipage}
\hfill
\begin{minipage}[t]{0.325\linewidth}
\centering
\includegraphics[width=\textwidth]{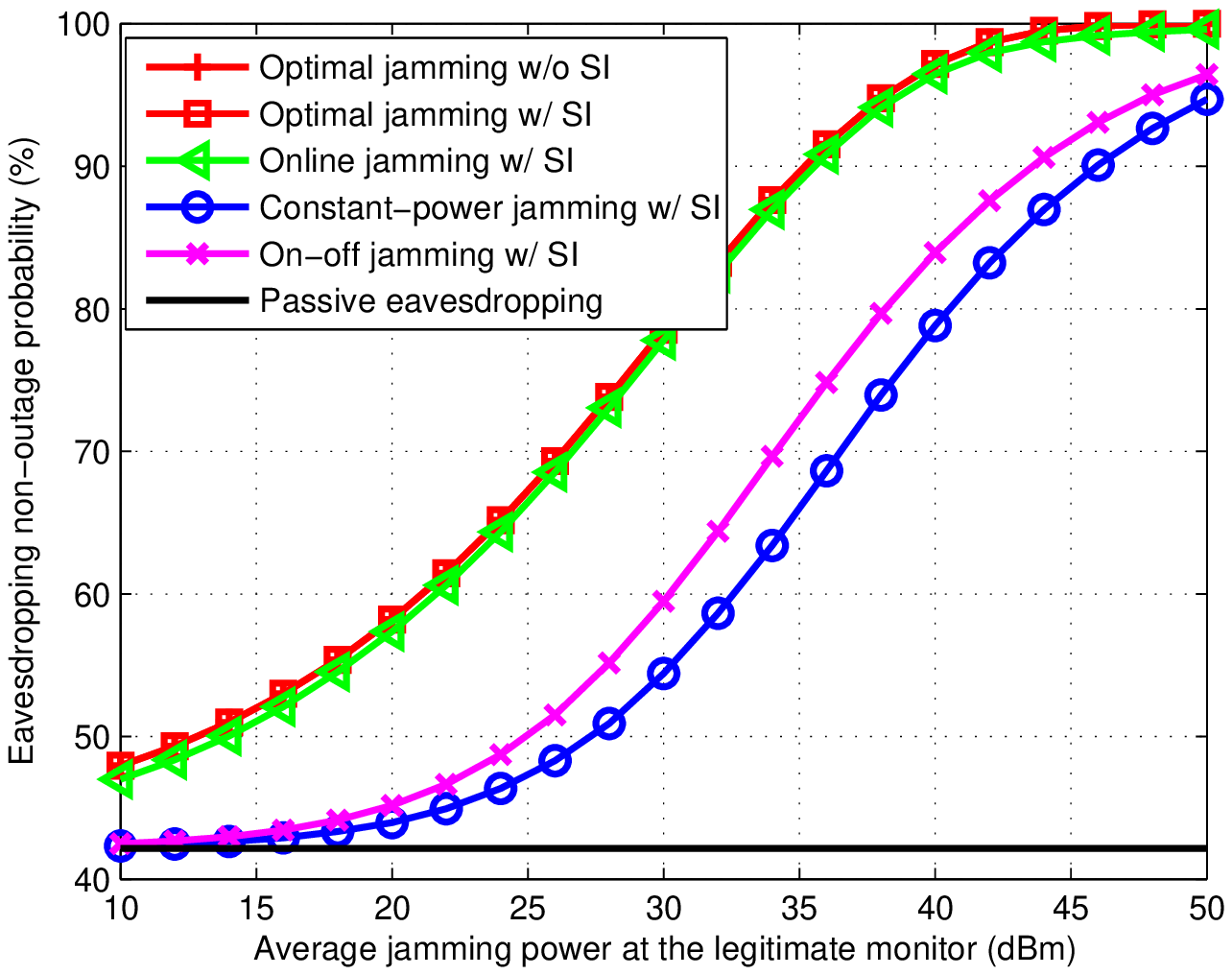}
\caption{Performance comparison between the optimal and online cognitive jamming in the case when the eavesdropping and jamming antennas are separately located at the legitimate monitor.} \label{fig:11}
\end{minipage}\vspace{-1em}
\end{figure*}

\subsection{Numerical Examples}

We conduct simulations to illustrate the effect of residual SI and show the performance of our proposed online cognitive jamming design under a practical setup with $N=10^5$ time blocks. In the simulations, we consider the suspicious transmitter and the suspicious receiver are located at (0, 0) and (500 meters, 0), respectively. We consider Rayleigh fading channel model, where the pathloss is assumed to be $\iota (d/d_0)^{-\kappa}$, with $\iota = -60$ dB at a reference distance of $d_0 = 10$ meters, and the pathloss exponent is $\kappa=3$. Here, $d$ denotes the distance between a transmitter and a receiver. Furthermore, we consider the SIC capability at the legitimate monitor to be $\varphi = 110$ dB \cite{fullduplex}. For the practical online cognitive jamming, we set the initial threshold as $\tau(1) = 2Q$, and the step size as $Q/1000$. In addition, we set the noise powers as $\sigma_0^2 = \sigma_1^2 = -80$ dBm, and the transmit power at the suspicious transmitter as $P=40$ dBm.

First, consider that the jamming and eavesdropping antennas of the legitimate monitor are co-located at (500 meters, 500 meters), where the loop-back channel power gain is assumed to be $\tilde \phi(\nu) = -15$ dB, $\forall \nu$ (with the distance between the eavesdropping and jamming antennas being a half wavelength) \cite{LoopBackChannel}. Fig. \ref{fig:9} shows the thresholds $1/\lambda_4^*$ by the optimal jamming and $\tau(\nu)$ by the practical online jamming, where $Q = 30$ dBm. It is observed that the threshold under practical online jamming converges to a similar value as the one under the optimal jamming, though it fluctuates over time due to the relatively large step size employed. Fig. \ref{fig:10} shows the eavesdropping non-outage probability versus the average jamming power $Q$. It is observed that our proposed online jamming (with SI) achieves close performances to the optimal jamming (with SI), which shows the effectiveness of our online threshold adaptation and jamming power design. The online jamming is also observed to significantly outperform other benchmark schemes including constant-power jamming, on-off jamming, and passive jamming. Furthermore, the performance achieved by the optimal jamming with SI is inferior to that without SI, especially when the average jamming power is larger than 30 dBm. Such a performance loss is due to the residual SI that is significant at the co-located legitimate monitor.

Next, consider that the eavesdropping and jamming antennas of the legitimate monitor are separately located at (250 meters, 500 meters) and (500 meters, 500 meters), respectively. Fig. \ref{fig:11} shows the eavesdropping non-outage probability versus the average jamming power $Q$. Due to the effectiveness of SIC in this case, the optimal jamming with SI is observed to perform the same as that without SI. Furthermore, the proposed online jamming with SI is observed to have a similar performance as that of optimal jamming, and achieves much higher eavesdropping non-outage probability than the other benchmark schemes.

%
%
%\begin{figure}
%\centering
% \epsfxsize=1\linewidth
%    \includegraphics[width=8cm]{threshold_colocated_P30dB.eps}
%\caption{The threshold comparison between the optimal and online cognitive jamming, where the maximum average jamming power is set as $Q = 30$ dBm.}\label{fig:9}
%\end{figure}
%
%\begin{figure}
%\centering
% \epsfxsize=1\linewidth
%    \includegraphics[width=8cm]{colocated.eps}
%\caption{Performance comparison between the optimal and online cognitive jamming in the case when the eavesdropping and jamming antennas are co-located at the legitimate monitor.}\label{fig:10}
%\end{figure}
%
%\begin{figure}
%\centering
% \epsfxsize=1\linewidth
%    \includegraphics[width=8cm]{separate.eps}
%\caption{Performance comparison between the optimal and online cognitive jamming in the case when the eavesdropping and jamming antennas are separately located at the legitimate monitor.} \label{fig:11}\vspace{-0em}
%\end{figure}

\section{Concluding Remarks}\label{sec:VII}

This paper proposes a new proactive eavesdropping via cognitive jamming approach for a legitimate monitor to efficiently intercept a point-to-point suspicious communication link in fading channels. Under ideal assumptions of perfect SIC and global CSI and by considering both delay-sensitive and delay-tolerant suspicious applications, we formulate optimization problems to maximize the eavesdropping non-outage probability and the relative eavesdropping rate at the legitimate monitor, respectively, by optimizing its jamming power allocation subject to an average jamming power constraint. Despite the non-convexity of these problems, we obtain their optimal solutions by utilizing the Lagrange duality method. Numerical results show that our proposed proactive eavesdropping via cognitive jamming can significantly improve the eavesdropping performance as compared to conventional heuristics. Our proposed proactive eavesdropping design is also extended to the practical case with residual SI and local CSI only. We hope that this paper can provide a new paradigm for designing legitimate surveillance in emerging wireless communication networks. Due to the space limitation, there are various important issues that are unaddressed in this paper. We briefly discuss them in the following to motivate future studies.

First, in the future the suspicious users may be intelligent and be able to detect the legitimate monitor (see, e.g., \cite{MukherjeeICASSP2012}), deploy more antennas, and even utilize advanced physical-layer security techniques (aided by the artificial noise \cite{Goel2008}) to defend against the eavesdropping attack. These anti-eavesdropping techniques can be viewed as the countermeasure of the wireless information surveillance. Modeling and analyzing their interplay, e.g., via game theory \cite{HanZhu2011}, are interesting open problems.

Next, in practical wireless networks there may exist massive suspicious users each with more than one antennas, and they may adapt the transmit beamformers to defend against the eavesdropping. To ensure the successful eavesdropping in this case, we may need a large number of multi-antenna legitimate monitors with either separate or co-located eavesdropping/jamming antennas. How to select the mode (i.e., eavesdropping or jamming) for each antenna at different legitimate monitors, and coordinate the eavesdropping and jamming design at different antennas is an interesting problem worth pursuing in the future work.

Furthermore, to approach the proactive eavesdropping performance upper bound (beyond the online jamming), it is critical for the legitimate monitor to obtain the global CSI (especially the CSI of the suspicious link). Some channel learning ideas in cognitive radio and energy-based feedback (see, e.g., \cite{Zhang2010,NoamGoldsmith2013,XuZhang2014,GopalakrishnanSidiropoulos2014}) may be borrowed for the legitimate monitor to learn the CSI of the suspicious link.

\appendix
\subsection{Proof of Proposition \ref{proposition:P1}}\label{app:P1}
We prove Proposition \ref{proposition:P1} by using the Lagrange duality method. Let $\lambda \ge 0$ denote the dual variable associated with the average jamming power constraint in (\ref{eqn:P2:con1}). Then the partial Lagrangian of problem (P1) is expressed as
\begin{align}
{\mathcal L}_1(\{q(\nu)\},\lambda) = \mathbb{E}_{\nu}(X(\nu)) - \lambda\left(\mathbb{E}_{\nu}(q(\nu)) - Q\right).
\end{align}
Define the dual function as
\begin{align}
f_1(\lambda)=\max_{\{q(\nu)\ge 0\}}~\mathcal{L}_1(\{q(\nu)\},\lambda).\label{eqn:dual:function:P1}
\end{align}
Accordingly, the dual problem of (P1) is given by
\begin{align}
{\rm (D1)}: ~\min_{\lambda\ge 0} f_1(\lambda).
\end{align}

Since strong duality holds between (P1) and its dual problem (D1), we solve (P1) by equivalently solving (D1). In particular, we first solve problem (\ref{eqn:dual:function:P1}) to obtain $f_1(\lambda)$ under any given $\lambda$ and then solve problem (D1) to find the optimal $\lambda$, denoted by $\lambda_1^*$.

First, consider problem (\ref{eqn:dual:function:P1}) under any given $\lambda \ge 0$. By discarding the constant term $\lambda Q$, problem (\ref{eqn:dual:function:P1}) can be decomposed into a sequence of subproblems as follows each for one fading state $\nu$.
\begin{align}
\max_{q(\nu)\ge 0}X(\nu) - \lambda q(\nu)\label{problem:subproblem:dual:P1}
\end{align}
We solve problem (\ref{problem:subproblem:dual:P1}) by considering the two cases when $\left(\frac{g_0(\nu)}{g_1(\nu)}\sigma^2_1 - \sigma^2_0\right)\frac{1}{g_2(\nu)} \le 0$ and $\left(\frac{g_0(\nu)}{g_1(\nu)}\sigma^2_1 - \sigma^2_0\right)\frac{1}{g_2(\nu)} > 0$, respectively. When $\left(\frac{g_0(\nu)}{g_1(\nu)}\sigma^2_1 - \sigma^2_0\right)\frac{1}{g_2(\nu)} \le 0$, it always holds that $X(\nu) = 1$ provided that $q(\nu) \ge 0$, and thus problem (\ref{problem:subproblem:dual:P1}) becomes $\max_{q(\nu) \ge 0}~1 - \lambda q(\nu)$, for which the optimal solution is $q_1^{(\lambda)} = 0$.

On the other hand, when $\left(\frac{g_0(\nu)}{g_1(\nu)}\sigma^2_1 - \sigma^2_0\right)\frac{1}{g_2(\nu)} > 0$, problem (\ref{problem:subproblem:dual:P1}) is solved by comparing the optimal values under the following two subcases.

Subcase 1: $X(\nu) = 1$ or equivalently $q(\nu)\ge \left(\frac{g_0(\nu)}{g_1(\nu)}\sigma^2_1 - \sigma^2_0\right)\frac{1}{g_2(\nu)}$. In this case, problem (\ref{problem:subproblem:dual:P1}) becomes $\max_{q(\nu) \ge \left(\frac{g_0(\nu)}{g_1(\nu)}\sigma^2_1 - \sigma^2_0\right)\frac{1}{g_2(\nu)}}~1 - \lambda q(\nu)$, for which the solution is $q(\nu) = \left(\frac{g_0(\nu)}{g_1(\nu)}\sigma^2_1 - \sigma^2_0\right)\frac{1}{g_2(\nu)}$, and the resulting optimal value is $1 - \lambda \left(\frac{g_0(\nu)}{g_1(\nu)}\sigma^2_1 - \sigma^2_0\right)\frac{1}{g_2(\nu)}$.

Subcase 2: $X(\nu) = 0$ or equivalently $q(\nu)< \left(\frac{g_0(\nu)}{g_1(\nu)}\sigma^2_1 - \sigma^2_0\right)\frac{1}{g_2(\nu)}$. In this case, problem (\ref{problem:subproblem:dual:P1}) becomes $\max_{0 \le q(\nu) < \left(\frac{g_0(\nu)}{g_1(\nu)}\sigma^2_1 - \sigma^2_0\right)\frac{1}{g_2(\nu)}} - \lambda q(\nu)$, for which the solution is $q(\nu) = 0$, and the corresponding optimal value is $0$.

By comparing the two subcases, we have that if $1 - \lambda \left(\frac{g_0(\nu)}{g_1(\nu)}\sigma^2_1 - \sigma^2_0\right)\frac{1}{g_2(\nu)} > 0$, then $q_1^{(\lambda)} = \left(\frac{g_0(\nu)}{g_1(\nu)}\sigma^2_1 - \sigma^2_0\right)\frac{1}{g_2(\nu)}$; otherwise, $q_1^{(\lambda)} = 0$. By summarizing the above two cases, the optimal solution to problem (\ref{problem:subproblem:dual:P1}) is given as
\begin{align}
&q_1^{(\lambda)}(\nu) = \nonumber\\& \left\{
\begin{array}{ll}
\left(\frac{g_0(\nu)}{g_1(\nu)}\sigma^2_1 - \sigma^2_0\right)\frac{1}{g_2(\nu)}, & {\rm if}~0 < \left(\frac{g_0(\nu)}{g_1(\nu)}\sigma^2_1 - \sigma^2_0\right)\frac{1}{g_2(\nu)} < \frac{1}{\lambda}, \\
0, &{\rm otherwise}.
\end{array}
\right.\label{eqn:q1nu}
\end{align}
Therefore, the dual function $f_1(\lambda)$ has been obtained.

Next, we solve the dual problem (D1) to find the optimal $\lambda_1^*$ via the bisection method by using the fact that the subgradient of $f_1(\lambda)$ is indeed $s_1(\lambda) = Q-\mathbb{E}_{\nu}(q_1^{(\lambda)}(\nu))$ under any given $\lambda \ge 0$. By substituting the optimal $\lambda_1^*$ into (\ref{eqn:q1nu}), then the optimal solution to (P1) is given as $\{q_1^*(\nu)\}$ in (\ref{eqn:optimal:P1}). Note that if $\mathbb{E}_{\nu}\left(\left(\frac{g_0(\nu)}{g_1(\nu)}\sigma^2_1 - \sigma^2_0\right)\frac{1}{g_2(\nu)}\right) < Q$, then we have $s_1(\lambda) > 0, \forall \lambda \ge 0$. In this case, we have $\lambda^* \to 0$ and the optimal solution degrades to (\ref{eqn:optimal:P1:case1}). Otherwise, $\lambda_1^*$ is set such that $s_1(\lambda_1^*) = Q - \mathbb{E}_{\nu}(q_1^*(\nu)) = 0$. Therefore, the proposition is finally proved.
\vspace{-0em}
\subsection{Proof of Proposition \ref{proposition:P2:1}}\label{app:A}\vspace{-0em}
First, we prove the `if' part. Let $\{q(\nu)\}$ be a feasible solution set, then for any $\mu \ge 0$ and $\lambda \ge 0$, it follows that
%\begin{align*}
$f_2(\mu,\lambda) \ge \mathcal{L}_2(\{q(\nu)\},\mu,\lambda) \ge 0.$
%\end{align*}
Then if there exist $\mu \ge 0$ and $\lambda \ge 0$ such that $f_2(\mu,\lambda) < 0$, then problem (P2.2) is infeasible and $\mu$ and $\lambda$ are one certificate of infeasibility.

Next, we prove the `only if' part. Consider a given $\lambda > 0$, and define the following problem.
\begin{align}
\max_{\{q(\nu)\}}~&~\lambda\left(Q - \mathbb{E}_{\nu}(q(\nu))\right)  \label{eqn:aux:problem}\\
{\rm s.t.}~&~(\ref{eqn:P21:con1})~{\mathrm{and}}~(\ref{eqn:P2:con2})\nonumber
\end{align}
Note that problem (\ref{eqn:aux:problem}) is always feasible for any $0\le t\le 1$, via the legitimate monitor setting its jamming power to be sufficiently large, e.g., $q(\nu) \to \infty, \forall \nu$. Let the optimal solution to problem (\ref{eqn:aux:problem}) be denoted by $\{\underline q(\nu)\}$. Since problem (P2.2) is infeasible, it follows that $\mathbb{E}_{\nu}(\underline q(\nu)) > Q$ and thus $\lambda\left(Q- \mathbb{E}_{\nu}(\underline q(\nu))\right) < 0$. Furthermore, note that strong duality holds for problem (\ref{eqn:aux:problem}) since it satisfies the time-sharing condition. By letting $\mu$ denote the dual variable associated with the constraint (\ref{eqn:P21:con1}) in problem (\ref{eqn:aux:problem}), then it is easy to show that there exists a dual variable $\underline{\mu} \ge 0$ such that
%\begin{align*}
$\max_{\{q(\nu)\ge 0\}}~\mathcal{L}_2(\{q(\nu)\},\underline{\mu},\lambda) = \mathcal{L}_2(\{\underline q(\nu)\},\underline{\mu},\lambda) < 0.$
%\end{align*}
As a consequence, we have $f_2(\underline{\mu},\lambda) = \mathcal{L}_2(\{\underline q(\nu)\},\underline{\mu},\lambda)  < 0$. Equivalently, there exist $\mu \ge 0$ and $\lambda \ge 0$ such that $f_2(\mu,\lambda) < 0$. Therefore, this proposition follows immediately.
\vspace{-0em}
\subsection{Proof of Proposition \ref{proposition:P2:2}}\label{app:B}\vspace{-0em}
First, we consider the case when $\left(\frac{g_0(\nu)}{g_1(\nu)}\sigma^2_1 - \sigma^2_0\right)\frac{1}{g_2(\nu)} \le 0$. In this case, it always holds that $X(\nu) = 1$ provided that $q(\nu) \ge 0$. As a result, problem (\ref{eqn:P2:dual:function:decom}) becomes
\begin{align*}
\max_{q(\nu)\ge 0}~\mu(1-t)r_0(\nu) - \lambda q(\nu),
\end{align*}
for which the optimal solution is $q_2^{(\mu,\lambda)}(\nu) = 0$.

Next, we consider the other case when $\left(\frac{g_0(\nu)}{g_1(\nu)}\sigma^2_1 - \sigma^2_0\right)\frac{1}{g_2(\nu)} > 0$. In this case, problem (\ref{eqn:P2:dual:function:decom}) is solved by comparing the optimal values under the two subcases when $X(\nu) = 1$ and $X(\nu) = 0$, respectively.
%\begin{itemize}
%  \item

Subcase 1: $X(\nu) = 1$ or equivalently $q(\nu)\ge \left(\frac{g_0(\nu)}{g_1(\nu)}\sigma^2_1 - \sigma^2_0\right)\frac{1}{g_2(\nu)}$. In this subcase, problem (\ref{eqn:P2:dual:function:decom}) becomes $\max_{q(\nu)\ge \left(\frac{g_0(\nu)}{g_1(\nu)}\sigma^2_1 - \sigma^2_0\right)\frac{1}{g_2(\nu)}}\mu\left(1- t\right) r_0(\nu) - \lambda q(\nu)$, for which the solution is $q(\nu) = \left(\frac{g_0(\nu)}{g_1(\nu)}\sigma^2_1 - \sigma^2_0\right)\frac{1}{g_2(\nu)}$, and the resultant optimal value is given as $v_1(\nu)$ in (\ref{eqn:v1}).

Subcase 2: $X(\nu) = 0$ or equivalently $q(\nu)< \left(\frac{g_0(\nu)}{g_1(\nu)}\sigma^2_1 - \sigma^2_0\right)\frac{1}{g_2(\nu)}$. In this subcase, problem (\ref{eqn:P2:dual:function:decom}) becomes \begin{align}
      \max_{q(\nu)}&~ - \mu t r_0(\nu) - \lambda q(\nu)\nonumber\\
      \mathrm{s.t.}&~0 \le q(\nu) < \left(\frac{g_0(\nu)}{g_1(\nu)}\sigma^2_1 - \sigma^2_0\right)\frac{1}{g_2(\nu)}.\label{eqn:case2}
      \end{align}
      Note that $r_0(\nu)$ is a convex function in $q(\nu) \ge 0$. As a result, problem (\ref{eqn:case2}) is a convex optimization problem. By using the standard convex optimization technique, the optimal solution to problem (\ref{eqn:case2}) is given as $q(\nu) = \bar q(\nu)$ in (\ref{eqn:hatq}) and the resulting optimal value is expressed as $v_2(\nu)$ in (\ref{eqn:v2}). Note that in the case with $\bar q(\nu) = \left(\frac{g_0(\nu)}{g_1(\nu)}\sigma^2_1 - \sigma^2_0\right)\frac{1}{g_2(\nu)}$, the solution $q(\nu) = \bar q(\nu)$ here cannot be exactly achieved due to the strict power inequality constraint in problem (\ref{eqn:case2}). Nevertheless, this would not affect the solution to (\ref{eqn:P2:dual:function:decom}), since in this case the optimal value $v_2(\nu)$ here is always smaller than that in the subcase 1, i.e., $v_1(\nu)$.

By comparing the optimal values $v_1(\nu)$ and $v_2(\nu)$, the optimal solution $q_2^{(\mu,\lambda)}(\nu)$ in the case when $\left(\frac{g_0(\nu)}{g_1(\nu)}\sigma^2_1 - \sigma^2_0\right)\frac{1}{g_2(\nu)} > 0$ can be obtained. By using this together with the solution $q_2^{(\mu,\lambda)}(\nu) = 0$ in the case with $\left(\frac{g_0(\nu)}{g_1(\nu)}\sigma^2_1 - \sigma^2_0\right)\frac{1}{g_2(\nu)} \le 0$, the optimal solution to problem (\ref{eqn:P2:dual:function:decom}) is finally given in (\ref{eqn:dual:function:P2:solution}). Therefore, this proposition is proved.

\subsection{Checking the Feasibility of Problem (\ref{eqn:29:feasibility})}\label{app:D}
Let the dual variables associated with the constraints in (\ref{eqn:P2:con1}) and (\ref{eqn:Pave}) be denoted by $\lambda \ge 0$ and $\zeta$, respectively. Then the partial Lagrangian of problem (\ref{eqn:29:feasibility}) is denoted as
\begin{align}
&\mathcal{\hat{L}}(\{q(\nu)\},\lambda,\zeta) = - \lambda \left(\mathbb{E}_{\nu}(q(\nu)) - Q\right) - \zeta\left(\mathbb{E}_\nu\left(\hat p(\nu)\right) - P\right).
\end{align}
As a result, the dual function of (P2.2) is expressed as
\begin{align}
\hat f(\lambda,\zeta)=\max_{\{q(\nu)\ge 0\}}~\mathcal{\hat{L}}(\{q(\nu)\},\lambda,\zeta).\label{eqn:feasibility}
\end{align}
The dual problem is accordingly written as $\min_{\lambda\ge 0,\zeta}~\hat f(\lambda,\zeta)$.

Based on Lemma \ref{lemma:2}, we check the feasibility of problem (\ref{eqn:29:feasibility}) by solving its dual problem. Similar to Proposition \ref{proposition:P2:1}, we have the following proposition, for which the proof is omitted for brevity.
\begin{proposition}\label{proposition:3}
Problem (\ref{eqn:29:feasibility}) is infeasible if and only if there exist $\lambda \ge 0$ and $\zeta$ such that $\hat f(\lambda,\zeta) < 0$.
\end{proposition}

In addition, we have the optimal solution to problem (\ref{eqn:feasibility}) given in the following proposition.
\begin{proposition}\label{proposition:4}
The optimal solution to problem (\ref{eqn:feasibility}) is given as
\begin{align}\label{eqn:opt:29:feasibility}
\underline{q}^{(\lambda,\zeta)}(\nu) = \left\{
\begin{array}{ll}
\hat q_1(\nu),& ~{\rm if}~ -  \frac{\zeta}{\ln 2 \cdot \beta}+ \frac{\zeta\sigma^2_0}{g_0(\nu)} \le  - \lambda \hat q_1(\nu)\\
0,& ~{\rm otherwise},
\end{array}
\right.
\end{align}
where $\hat q_1(\nu) \triangleq \left[\frac{g_0(\nu)}{\ln2\cdot \beta g_2(\nu)} - \frac{\sigma_0^2}{g_2(\nu)}\right]^+, \forall \nu$.
\end{proposition}
\begin{IEEEproof}
Problem (\ref{eqn:feasibility}) can be decomposed into various subproblems as follows each for one fading state $\nu$.
\begin{align}
\max_{q(\nu)\ge 0}~- \lambda q(\nu) - \zeta\hat p(\nu)\label{eqn:feasibility:decomposed}
\end{align}
When $q(\nu) \ge \hat q_1(\nu)$, problem (\ref{eqn:feasibility:decomposed}) becomes $\max_{q(\nu)\ge \hat q_1(\nu)}~- \lambda q(\nu)$, for which the optimal solution is $q(\nu) = \hat q_1(\nu)$ and the resulting optimal value is $- \lambda \hat q_1(\nu)$. On the other hand, when $q(\nu) < \hat q_1(\nu)$, problem (\ref{eqn:feasibility:decomposed}) becomes $\max_{0 \le q(\nu)<\hat q_1(\nu)}~ - \lambda q(\nu) - \zeta \left(\frac{1}{\ln 2 \cdot \beta} - \frac{g_2(\nu)q(\nu)+\sigma^2_0}{g_0(\nu)} \right)$. The objective values under $q(\nu) = 0$ and $q(\nu) = \hat q_1(\nu)$ are given as $- \zeta \frac{1}{\ln 2 \cdot \beta} + \frac{\zeta\sigma^2_0}{g_0(\nu)}$ and $- \lambda \hat q_1(\nu)$, respectively. As a result, by comparing them, we have the optimal solution to problem (\ref{eqn:feasibility:decomposed}) as given in (\ref{eqn:opt:29:feasibility}). Therefore, Proposition \ref{proposition:4} is verified.
\end{IEEEproof}

Based on Propositions \ref{proposition:3} and \ref{proposition:4}, we can efficiently check the feasibility of problem (\ref{eqn:29:feasibility}) by using the ellipsoid method by using the fact that the subgradient of $\hat f(\lambda,\zeta)$ is given by $\hat{\mv s}(\lambda,\zeta) =   \left[Q - \mathbb{E}_{\nu}(\underline{q}^{(\lambda,\zeta)}(\nu)),~P-\mathbb{E}_\nu\left(\underline{\hat p}^{(\lambda,\zeta)}(\nu)\right)\right]^T$, where $\{\underline{\hat p}^{(\lambda,\zeta)}(\nu)\}$ denotes the corresponding $\{\hat p(\nu)\}$ in (\ref{eqn:solution:p}) under given $\{\underline{q}^{(\lambda,\zeta)}(\nu)\}$.

\subsection{Proof of Proposition \ref{proposition:5}}\label{app:E}

Note that by discarding the constant $\lambda Q + \zeta P$, problem (\ref{eqn:P2:2:dual:function}) can be equivalent decomposed into various subproblems in the following, each of which is for one fading state $\nu$.
\begin{align}
\max_{q(\nu)\ge 0}~ &\mu(X(\nu)-t)\hat r_0(\nu)  - \lambda q(\nu) - \zeta\hat p(\nu)\label{eqn:P2:2:dual:function:state}
\end{align}
For each fading state $\nu$, problem (\ref{eqn:P2:2:dual:function:state}) is solved by considering three cases.

Consider the first case when $q(\nu) \ge \hat q_1(\nu)$, in which we have $\hat r_0(\nu) = 0$ and $\hat p(\nu) = 0$. Accordingly, problem (\ref{eqn:P2:2:dual:function:state}) becomes $\max_{q(\nu)\ge \hat q_1(\nu)}~ - \lambda q(\nu)$, for which the optimal solution and the resulting optimal value are $q(\nu) = \hat q_1(\nu)$ in (\ref{eqn:hatq1}) and $\hat v_1(\nu)$ in (\ref{eqn:hatv1}), respectively.

Next, consider the second case with $q(\nu) < \hat q_1(\nu)$ and $X(\nu) = 1$ (or equivalently $q(\nu)\ge \hat q_2(\nu)$ with $\hat q_2(\nu)$ given in (\ref{eqn:hatq2})). In this case, problem (\ref{eqn:P2:2:dual:function:state}) becomes
\begin{align}
\max_{q(\nu)}~ &\mu(1-t)\log_2\left( \frac{g_0(\nu)}{\ln 2 \cdot \beta(g_2(\nu)q(\nu) + \sigma^2_0)} \right)  - \lambda q(\nu) \nonumber\\& - \zeta \left(\frac{1}{\ln 2 \cdot \beta} - \frac{g_2(\nu)q(\nu)+\sigma^2_0}{g_0(\nu)} \right)\nonumber\\
\mathrm{s.t.}~& \hat q_2(\nu)\le q(\nu)<\hat q_1(\nu).\label{eqn:problem:subcase2}
\end{align}
Note that problem (\ref{eqn:problem:subcase2}) is feasible only when $\hat q_2(\nu)<\hat q_1(\nu)$. Furthermore, the objective function of problem (\ref{eqn:problem:subcase2}) is convex as a function of $q(\nu)$. As a result, its optimal solution is either $q(\nu) = \hat q_2(\nu)$ or $q(\nu) = \hat q_1(\nu)$. When $q(\nu) = \hat q_1(\nu)$, the objective value is $\hat v_1(\nu)$ in (\ref{eqn:hatv1}), while when $q(\nu) = \hat q_2(\nu)$, the objective value is $\hat v_2(\nu)$ in (\ref{eqn:hatv2}).

In addition, consider the third case with $q(\nu) < \hat q_1(\nu)$ and $X(\nu) = 0$ (or equivalently $q(\nu) < \hat q_2(\nu)$. In this case, problem (\ref{eqn:P2:2:dual:function:state}) becomes
\begin{align}
\max_{q(\nu)}~ &-t\mu\log_2\left( \frac{g_0(\nu)}{\ln 2 \cdot \beta(g_2(\nu)q(\nu) + \sigma^2_0)} \right)  - \lambda q(\nu) \nonumber\\&- \zeta \left(\frac{1}{\ln 2 \cdot \beta} - \frac{g_2(\nu)q(\nu)+\sigma^2_0}{g_0(\nu)} \right)\nonumber\\
\mathrm{s.t.}~& 0\le q(\nu) < \min(\hat q_1(\nu),\hat q_2(\nu)),\label{eqn:pro:40}
\end{align}
which is a convex optimization problem. It can be shown that the first-order derivative of the objective function of problem (\ref{eqn:pro:40}) achieves zero value when $\hat q_4(\nu) $ in (\ref{eqn:q4}). As a result, the optimal solution to problem (\ref{eqn:pro:40}) is given as $\hat q_3(\nu) = \left[\min(\hat q_1(\nu),\hat q_2(\nu),\hat q_4(\nu))\right]^+$ as given in (\ref{eqn:hatq3})), and the resulting optimal value is given as $\hat v_3(\nu)$ in (\ref{eqn:hatv3}).

By comparing the obtained values $\hat v_1(\nu)$, $\hat v_2(\nu)$, and $\hat v_3(\nu)$ in the above three cases, together with the fact that $\hat v_2(\nu)$ is achievable (i.e., problem (\ref{eqn:problem:subcase2}) is feasible) only when $\hat q_2(\nu)<\hat q_1(\nu)$, we can obtain the optimal solution to problem (\ref{eqn:P2:2:dual:function:state}). Therefore, this proposition is proved.

\end{document}